\documentclass[10pt,journal,letterpaper,romanappendices]{IEEEtran}

\usepackage[english]{babel} 
\usepackage{amsmath}                
\usepackage{amsfonts}               
\usepackage{amssymb}                
\usepackage{amsopn}                 
\usepackage{bbm}                    
\usepackage{mathrsfs}               
\usepackage{calc}                   
\usepackage[dvips]{graphicx}        
\usepackage{epsfig}                
\usepackage{psfrag}                 
\usepackage[dvips]{color}           
\usepackage{fancyhdr}               
\usepackage{verbatim}               
\usepackage{exscale}                
\usepackage{macros}

\newtheorem{theorem}{Theorem}

\newtheorem{corollary}{Corollary}
\newtheorem{proposition}{Proposition}

\title{On Noncoherent Fading Relay Channels\\ at High Signal-to-Noise Ratio}

\author{Tobias~Koch,~\IEEEmembership{Member,~IEEE,}~and~Gerhard Kramer,~\IEEEmembership{Fellow,~IEEE}%
\thanks{T.~Koch has received funding from the European Community's Seventh Framework Programme (FP7/2007-2013) under grant agreement No. 252663 and from the Ministerio de Econom\'ia of Spain (projects DEIPRO, id.~TEC2009-14504-C02-01, and COMONSENS, id.~CSD2008-00010). G.~Kramer was supported by an Alexander von Humboldt Professorship endowed by the German Federal Ministry of Education and Research. He was also supported by ARO Grant W911NF-06-1-0182 and by NSF Grant CCF-09-05235. The material in this paper was presented in part at the 2005 IEEE International Symposium on Information Theory (ISIT), Adelaide, Australia, Sept.~4--9, 2005.}
 \thanks{T.~Koch was with the Department of Engineering, University of Cambridge, Cambridge CB2 1PZ, UK. He is now with the Signal Theory and Communications Department, Universidad Carlos III de Madrid, 28911 Legan\'es, Spain (koch@tsc.uc3m.es).}
 \thanks{G.~Kramer is with the Institute for Communications Engineering, Technische Universit\"at M\"unchen, D-80333 M\"unchen, Germany (e-mail: gerhard.kramer@tum.de).}
}

\begin{document}

\maketitle

\begin{abstract}
The capacity of noncoherent fading relay channels is studied where all terminals are aware of the fading statistics but not of their realizations. It is shown that if the fading coefficient of the channel between the transmitter and the receiver can be predicted more accurately from its infinite past than the fading coefficient of the channel between the relay and the receiver, then at high signal-to-noise ratio (SNR) the relay does not increase capacity. It is further shown that if the fading coefficient of the channel between the transmitter and the relay can be predicted more accurately from its infinite past than the fading coefficient of the channel between the relay and the receiver, then at high SNR one can achieve communication rates that are within one bit of the capacity of the multiple-input single-output fading channel that results when the transmitter and the relay can cooperate.
\end{abstract}

\section{Introduction}
\label{sec:intro} 
\IEEEPARstart{A}{} \emph{relay channel} consists of a transmitter, a receiver, and a relay that supports the transmitter in communicating with the receiver. We study \emph{fading} relay channels, where the word ``fading'' refers to the variation in the strength of the links between the terminals. Coherent fading relay channels were studied, e.g., in \cite{kramergastpargupta05}, \cite{wangzhanghostmadsen05}. For such channels, the fading coefficients are available at the corresponding receiving terminals.

The assumption that the fading coefficients are available at the receiving terminals is commonly justified by saying that these coefficients vary slowly over time and can therefore be estimated by transmitting training sequences. However, this assumption yields overly-optimistic results, since it is \emph{prima facie} not clear whether the fading coefficients can be estimated perfectly, and since the transmission of training sequences reduces the achievable communication rates. For instance, in the point-to-point case (where a transmitter communicates with a receiver without the aid of a relay) the loss in not knowing the fading coefficient at the receiver can be substantial. Indeed, if the fading is \emph{regular} in the sense that the present fading coefficient cannot be predicted perfectly from its infinite past, then at high signal-to-noise ratio (SNR) the capacity grows double-logarithmically with the SNR \cite{lapidothmoser03_3} which is in stark contrast to the logarithmic growth in the coherent case \cite{ericson70}. If the fading is \emph{nonregular}, then the capacity can grow logarithmically with the SNR, but the pre-log, defined as the limiting ratio of capacity to $\log\SNR$ as $\SNR$ tends to infinity, depends on the fading's autocovariance function and is typically strictly smaller than one \cite{lapidoth05}.

In this paper, we study the capacity of noncoherent fading relay channels with \emph{regular} fading. For such channels the terminals are aware of the laws of the fading coefficients but not of their realizations. We derive two basic results. First, we show that if the fading coefficient of the channel between the transmitter and the receiver can be predicted more accurately from its infinite past than the fading coefficient of the channel between the relay and the receiver, then at high SNR the relay does not increase capacity. Second, we show that if the fading coefficient of the channel between the transmitter and the relay can be predicted more accurately from its infinite past than the fading coefficient of the channel between the relay and the receiver, then at high SNR one can achieve communication rates that are within one bit of the capacity of the multiple-input single-output (MISO) fading channel that results when the transmitter and the relay can cooperate. Thus, at high SNR the rate penalty for establishing cooperation between the transmitter and the relay is not greater than one bit.

We model the fading coefficients as stationary and ergodic stochastic processes whose autocovariance functions determine the fading's time-variation. This excludes the nonstationary \emph{block-fading model} introduced by Marzetta and Hochwald \cite{marzettahochwald99}. In the point-to-point case, the block-fading model and the stationary and ergodic fading model yield very different capacity behaviors at high SNR \cite{lapidothmoser03_3,lapidoth05,zhengtse02,liangveeravalli04}. This will also be the case for fading relay channels.

This paper is organized as follows. Section~\ref{sec:channel} describes the channel model. Section~\ref{sec:capacity} introduces channel capacity and defines the fading number. Section~\ref{sec:results} presents the main results. Section~\ref{sec:nonasymptotic} presents nonasymptotic bounds on the capacity of the fading relay channel. Sections~\ref{sec:upperbound} and \ref{sec:lowerbound} contain the proofs of the main results. Sections~\ref{sec:discussion} and \ref{sec:conclusion} conclude the paper with a discussion and summary of the obtained results.

\section{Channel Model}
\label{sec:channel}
\begin{figure}[t!]
\centering
 \epsfig{file=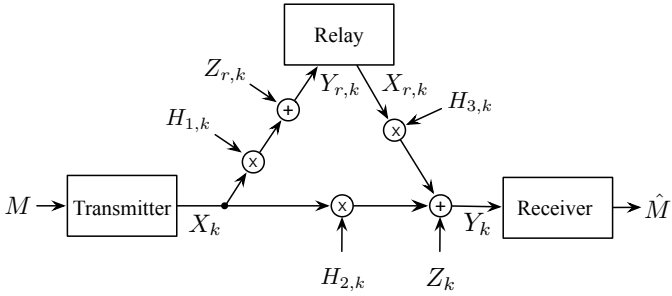, width=0.49\textwidth}
 \caption{The fading relay channel.}
 \label{fig1}
\end{figure}
The fading relay channel has three terminals (see Figure~\ref{fig1}): the transmitter, the receiver, and the relay. The message $M$ is uniformly distributed over the set \mbox{$\set{M}=\{1,\ldots,|\set{M}|\}$}, where $|\set{M}|$ is a positive integer.  The transmitter maps $M$ to the length-$n$ sequence $X_1^n$, where we use the short form $A_m^n$ to denote the sequence $A_m,\ldots,A_n$. Thus, we have $X_1^n=\phi_n\bigl(M\bigr)$ for some mapping $\phi_n\colon \set{M} \to \Complex^n$ (where $\Complex$ denotes the set of complex numbers). At each time instant $k\in\Integers$ (where $\Integers$ denotes the set of integers) the relay observes $Y_{r,k}\in\Complex$ and emits the symbol $X_{r,k}\in\Complex$, which is a function of the previously received symbols $Y_{r,1}^{k-1}$, i.e., $X_{r,k}=\varphi_{n,k}\bigl(Y_{r,1}^{k-1}\bigr)$, $k=1,\ldots,n$, for some mapping $\varphi_{n,k}\colon \Complex^{k-1} \to \Complex$. The receiver observes the channel output symbols $Y_1^n$ from which it guesses $M$. The receiver's guess is $\hat{M} = \psi_n\bigl(Y_1^n\bigr)$ for some mapping $\psi\colon \Complex^n \to \set{M}$.

The time-$k$ channel outputs $Y_{r,k}$ and $Y_k$ corresponding to the channel inputs $x_{k}$ and $x_{r,k}$ are given by
\begin{IEEEeqnarray}{rCll}
  Y_{r,k} & = & H_{1,k}x_k + Z_{r,k}, \quad & k\in\Integers\\
  Y_{k} & = & H_{2,k}x_k + H_{3,k}x_{r,k} + Z_k, \quad & k\in\Integers. 
\end{IEEEeqnarray}
Here $\{H_{1,k},\, k\in\Integers\}$, $\{H_{2,k},\, k\in\Integers\}$, $\{H_{3,k},\, k\in\Integers\}$, $\{Z_{r,k},\, k\in\Integers\}$, and $\{Z_{k},\, k\in\Integers\}$ are stationary and ergodic stochastic processes that take values in $\Complex$ and that are independent of each other. Furthermore, $\{H_{1,k},\,k\in\Integers\}$ and $\{Z_{r,k},\,k\in\Integers\}$ are of a joint law that does not depend on $\{x_{k},\,k\in\Integers\}$; and $\{H_{2,k},\,k\in\Integers\}$, $\{H_{3,k},\,k\in\Integers\}$, and $\{Z_{k},\,k\in\Integers\}$ are of a joint law that does not depend on $\{(x_{k},x_{r,k}),\,k\in\Integers\}$.

The additive noise terms $\{Z_k,\,k\in\Integers\}$ and $\{Z_{r,k},\,k\in\Integers\}$ are both sequences of independent and identically distributed (i.i.d.), zero-mean, circularly-symmetric, complex Gaussian random variables of variance $\sigma^2$. The multiplicative (fading) noise terms $\{H_{1,k},\,k\in\Integers\}$, $\{H_{2,k},\,k\in\Integers\}$, and $\{H_{3,k},\,k\in\Integers\}$ are zero-mean, unit-variance, stationary and ergodic, circularly-symmetric, complex Gaussian processes with the respective spectral distribution functions $F_1(\cdot)$, $F_2(\cdot)$, and $F_3(\cdot)$. Thus, $F_{\ell}(\cdot)$, $\ell=1,2,3$ are bounded and nondecreasing functions on $[-1/2,1/2]$ satisfying
\begin{equation}
\label{eq:spectraldist}
  \E{H_{\ell,k+m}\conj{H}_{\ell,k}} = \int_{-1/2}^{1/2} e^{\ii 2\pi m\lambda}\d F_{\ell}(\lambda), \quad \ell=1,2,3
\end{equation}
where $\ii=\sqrt{-1}$. We consider a \emph{noncoherent} channel model where the transmitter, receiver, and relay are not aware of the realizations of the fading processes $\{H_{\ell,k},\,k\in\Integers\}$, $\ell=1,2,3$ but only of their joint law. We assume that the fading processes $\{H_{\ell,k},\,k\in\Integers\}$, $\ell=1,2,3$ are regular in the sense that they satisfy
\begin{equation}
\label{eq:regularity}
\int_{-1/2}^{1/2} \log F'_{\ell}(\lambda)\d\lambda > -\infty, \quad \ell=1,2,3
\end{equation}
where $F'_{\ell}(\cdot)$ denotes the derivative of $F_{\ell}(\cdot)$. Note that $F_{\ell}(\cdot)$ is monotonic and almost everywhere differentiable. At the discontinuity points of $F_{\ell}(\cdot)$ the derivative $F'_{\ell}(\cdot)$ is undefined. If $F_{\ell}(\cdot)$ is absolutely continuous with respect to the Lebesgue measure on $[-1/2,1/2]$, then we shall refer to $F'_{\ell}(\cdot)$ as the \emph{spectral density} of $\{H_{\ell,k},\,k\in\Integers\}$. 

The assumption \eqref{eq:regularity} implies that the mean-square error in predicting $H_{\ell,0}$ from $H_{\ell,-1},H_{\ell,-2},\ldots$ is given by \cite{doob90}
\begin{equation}
\label{eq:epsregular}
\eps^2_{\ell} = \exp\biggl(\int_{-1/2}^{1/2} \log F'_{\ell}(\lambda)\d\lambda\biggr)
\end{equation}
and is strictly positive. We also have $\eps^2_{\ell}\leq 1$, $\ell=1,2,3$ since we take $\{H_{\ell,k},\,k\in\Integers\}$ to have unit variance. It follows from \eqref{eq:regularity} and \eqref{eq:epsregular} that a regular process cannot be predicted perfectly from its infinite past.

Processes with a bandlimited spectral density do not satisfy \eqref{eq:regularity} and are therefore \emph{nonregular}. (See \cite{kochkramer05} for results concerning the high-SNR capacity of noncoherent fading relay channels with nonregular fading.) Such processes can be predicted perfectly from their infinite past \cite[Sec.~10.1.5]{priestley81}, which leads to a dilemma. On the one hand, channel models based on the physics of the channel---such as Jakes' model \cite{jakes75}---suggest that practical fading processes have a bandlimited spectral density. On the other hand, such fading processes can be predicted perfectly from their infinite past, which seems unrealistic. Nevertheless, we believe that both regular and nonregular fading models are relevant, and that the answer to the question which one is more accurate depends on the SNR, bandwidth, and the channel statistics.

While fading processes with a bandlimited spectral density exhibit a direct relation between the spectral density's bandwidth and the Doppler spread of the channel, such a relation is less obvious for regular fading processes. A relationship between the spectral distribution function $F_{\ell}(\cdot)$ and the coherence time (which is inversely proportional to the Doppler spread) can be established by defining the coherence time as the time over which the autocorrelation function  $m\mapsto\Exp[H_{\ell,k+m}\conj{H}_{\ell,k}]$ is above, say, $1/2$ of its value at $0$, and by using \eqref{eq:spectraldist} to relate the autocorrelation function to $F_{\ell}(\cdot)$; see also \cite[Sec.~II]{etkintse06}. 

We assume that the channel inputs $X_k$ and $X_{r,k}$ satisfy \emph{peak-power constraints}, i.e., with probability one we have
\begin{IEEEeqnarray}{lCl}
  |X_k|^2 & \leq & \const{A}^2, \quad k\in\Integers  \label{eq:peakpower1} \\
  |X_{r,k}|^2 & \leq & \const{A}_r^2, \quad k\in\Integers   \label{eq:peakpower2}
\end{IEEEeqnarray}
for some positive real $\const{A}$ and $\const{A}_r$. We define
\begin{equation}
  \label{eq:alpha}
  \rho \triangleq \frac{\const{A}}{\const{A}_r}
\end{equation}
and
\begin{equation}
  \SNR \triangleq \frac{\const{A}^2}{\sigma^2}.
\end{equation}
Note that the main results presented in Section~\ref{sec:results} continue to hold if the peak-power constraints are replaced by average-power constraints.

\section{Channel Capacity and Fading Number}
\label{sec:capacity}
A \emph{rate} $R(\SNR,\rho)$ (in nats per channel use) is said to be \emph{achievable} if for every $\delta > 0$ there exists an $n>0$ and mappings $\phi_n$, $\bigl(\varphi_{n,1},\ldots,\varphi_{n,n}\bigr)$, and $\psi_n$ satisfying \eqref{eq:peakpower1} and \eqref{eq:peakpower2} such that
\begin{equation*}
  \frac{\log|\set{M}|}{n} > R(\SNR,\rho) - \delta
\end{equation*}
and such that the error probability satisfies $\Prob\bigl(\hat{M}\neq M\bigr)<\delta$. (Here $\log(\cdot)$ denotes the natural logarithm function.) The \emph{capacity} $C(\SNR,\rho)$ is defined as the supremum of all achievable rates.
We will focus on the asymptotic behavior of capacity at high SNR. For convenience, we assume that $\rho$ does not depend on SNR. This corresponds to the case where the available power at the relay is of the same order as the available power at the transmitter.

Let $C(\SNR)$ (without the parameter $\rho$) denote the capacity of the point-to-point channel. Lapidoth and Moser demonstrated that for regular fading, we have \cite[Th.~4.2]{lapidothmoser03_3}
\begin{equation}
\label{eq:loglogSNR}
  \varlimsup_{\SNR\to\infty} \bigl\{C(\SNR)-\log\log\SNR\bigr\} < \infty
\end{equation}
where $\varlimsup$ denotes the \emph{limit superior}.
They defined the \emph{fading number} $\chi$ as \cite[Def.~4.6]{lapidothmoser03_3}
\begin{equation}
\label{eq:fadingnumberdef}
  \chi \triangleq \varlimsup_{\SNR\to\infty} \bigl\{C(\SNR)-\log\log\SNR\bigr\}
\end{equation}
and computed its value for different fading channels. For instance, when the fading is a zero-mean, unit-variance, circularly-symmetric, complex Gaussian process with spectral distribution function $F(\cdot)$, the fading number is \cite[Cor.~4.42]{lapidothmoser03_3}
\begin{equation}
\label{eq:fadingnumberp2p}
  \chi = -1 - \gamma + \log\frac{1}{\eps^2}
\end{equation}
where $\gamma\approx 0.577$ denotes Euler's constant, and where $\eps^2$ denotes the mean-square error in predicting the present fading from its infinite past, given by \eqref{eq:epsregular}.

It follows from \eqref{eq:fadingnumberdef} that, at high SNR, the capacity can be approximated as
\begin{equation}
\label{eq:highSNRapprox}
C(\SNR) \approx \log\log\SNR + \chi.
\end{equation}
Thus communication is very power-inefficient at high SNR, since one should expect to square the $\SNR$ for every additional bit per channel use. For example, $\log\log\SNR$ is between $2.1$ and $3$ for $\SNR \in [30\text{dB}, 80\text{dB}]$ and the capacity can be approximated as
\begin{equation*}
2.1+\chi \lessapprox C(\SNR) \lessapprox 3 + \chi, \quad \SNR\in[30\text{dB},80\text{dB}].
\end{equation*}
This gives rise to the rule of thumb that a system operating at \emph{rates} considerably larger than $2+\chi$ operates in the high-SNR regime and is very power-inefficient \cite{lapidothmoser06}, see also \cite{lapidoth05,lapidoth03_2}. The fading number can therefore be viewed as an indication of the maximal rate up to which power-efficient communication is feasible.
However, it should be noted that it is difficult to determine the SNR at which this happens. Indeed, the fading number of zero-mean Gaussian fading channels depends on the spectral distribution function $F(\cdot)$ only via the mean-square error $\eps^2$ in predicting the present fading from its past, whereas the SNR at which \eqref{eq:highSNRapprox} becomes accurate is sensitive to the shape of $F(\cdot)$ \cite{lapidoth05,lapidoth03_2,kochlapidoth05_1}.\footnote{The SNR at which \eqref{eq:highSNRapprox} becomes accurate depends on the so-called \emph{noisy prediction error} \cite[Eq.~(11)]{lapidoth05}.}

Lapidoth and Moser prove \eqref{eq:loglogSNR} for multiple-input multiple-output noncoherent regular-fading channels \cite{lapidothmoser03_3}. It therefore follows from the \emph{max-flow min-cut upper bound} \cite[Th.~14.7.1]{coverthomas91} that at high SNR the capacity of the noncoherent fading relay channel also grows double-logarithmically with the SNR, implying that the power-inefficiency of communication at high SNR cannot be avoided by adding a relay. Nevertheless, the relay can increase the fading number, thereby increasing the maximal rate up to which power-efficient communication is feasible. In the following section, we present upper and lower bounds on the fading number of fading relay channels. They indicate by how much (if at all) adding a relay pushes the power-inefficient regime further away.

\section{Main Results}
\label{sec:results}
We define the fading number $\chi$ of the fading relay channel as in \eqref{eq:fadingnumberdef}, but with $C(\SNR)$ replaced by $C(\SNR,\rho)$. Note that our \emph{bounds} on $\chi$ do not depend on $\rho$, which is a consequence of the slow growth of the $\log\log(\cdot)$-function:
\begin{equation*}
\lim_{\const{A}\to\infty}\bigl\{\log\log\bigl(\rho\const{A}^2\bigr) - \log\log\const{A}^2\bigr\} = 0, \quad \rho>0.
\end{equation*}
We therefore do not make the dependence of $\chi$ on $\rho$ explicit in our notation.
An upper bound on the fading number follows from the \emph{max-flow min-cut upper bound}.
\begin{theorem}[Upper bound]
\label{thm:upperbound}
Consider the fading relay channel described in Section~\ref{sec:channel}. Assume that $\rho$ is independent of the SNR. Then we have
\begin{IEEEeqnarray}{lCll}
\chi & \leq & \min\Biggl\{& -2\gamma+\log\frac{1}{\eps^2_1}+\log\frac{1}{\eps^2_2},\nonumber\\
& & & \,\max\biggl\{-1-\gamma+\log\frac{1}{\eps^2_{2}},-1-\gamma+\log\frac{1}{\eps^2_3}\biggr\}\Biggr\} \IEEEeqnarraynumspace \label{eq:thmUB}
\end{IEEEeqnarray}
which for $\eps^2_2 \leq \eps^2_3$ becomes
\begin{equation}
\label{eq:thmUB2}
\chi \leq -1 -\gamma+\log\frac{1}{\eps^2_2}.
\end{equation}
The prediction errors $\eps^2_{\ell}$, $\ell=1,2,3$ are defined in \eqref{eq:epsregular}.
\end{theorem}
\begin{proof}
See Section~\ref{sec:upperbound}. Equation \eqref{eq:thmUB2} follows because $\eps^2_1\leq 1$ and because $-2\gamma > -1-\gamma$.
\end{proof}

Note that $-1-\gamma-\log \eps_2^2$ is the fading number of the channel between the transmitter and the receiver, while $-1-\gamma-\log\eps_3^2$ is the fading number of the channel between the relay and the receiver  \eqref{eq:fadingnumberp2p}. Thus, denoting the fading number of the former channel by $\chi_2$ and denoting the fading number of the latter channel by $\chi_3$, the upper bound \eqref{eq:thmUB} implies\begin{equation}
\label{eq:fadingnumberMISO}
\chi \leq \max\bigl\{\chi_2,\chi_3\bigr\}.
\end{equation}
The right-hand side (RHS) of \eqref{eq:fadingnumberMISO} is the fading number of a multiple-input single-output (MISO) fading channel with two transmit antennas, where the fading processes corresponding to the different antennas are independent, zero-mean, circularly-symmetric, complex Gaussian processes of spectral distribution function $F_{\ell}(\cdot)$, $\ell=2,3$ \cite{kochlapidoth05_3}, see also \cite{kochlapidoth05_1,koch04}. Thus, the fading number of the relay channel is upper-bounded by the fading number of the MISO channel that arises when the transmitter and the relay can cooperate. In the following, we shall refer to this channel as the \emph{TRC-MISO channel} (``TRC" stands for ``transmitter-relay cooperation").

It follows from \eqref{eq:fadingnumberMISO} that if the fading number of the channel between the transmitter and the receiver is larger than the fading number of the channel between the relay and the receiver, i.e., $\chi_2 \geq \chi_3$, then at high SNR the relay does not increase capacity:
\begin{corollary}
\label{cor:direct}
Let the fading processes $\{H_{2,k},\,k\in\Integers\}$ and $\{H_{3,k},\,k\in\Integers\}$ satisfy
\begin{equation*}
\eps^2_2 \leq \eps^2_3.
\end{equation*}
Then we have
\begin{equation}
\chi = -1-\gamma+\log\frac{1}{\eps^2_2}.
\end{equation}
\end{corollary}

Using a \emph{decode-and-forward} strategy \cite{coverelgamal79}, the following rates are achievable.
\begin{theorem}[Lower bound]
\label{thm:lowerbound1}
Consider the fading relay channel described in Section~\ref{sec:channel}. Assume that $\rho$ is independent of the SNR. Then we have
\begin{IEEEeqnarray}{lCll}
\chi & \geq & \max\Biggl\{& -1-\gamma+\log\frac{1}{\eps^2_2}, \nonumber\\
& & & \,{} -1-\gamma+\log\frac{1}{\eps^2_3}-\log\biggl(1+\frac{\eps^2_1}{\eps^2_3}\biggr)\Biggr\}.\label{eq:thmlowerbound1}
\end{IEEEeqnarray}
\end{theorem}
\begin{proof}
See Section~\ref{sec:lowerbound}.
\end{proof}

For $\eps_2^2 > \eps_1^2+\eps_3^2$, the RHS of \eqref{eq:thmlowerbound1} is strictly larger than \[-1-\gamma+ \log\frac{1}{\eps_2^2}.\] In this case, using a cooperative communication strategy rather than turning the relay off increases the fading number. We thus have the following corollary.
\begin{corollary}
\label{cor:cooperative}
Let the fading processes $\{H_{1,k}\,k\in\Integers\}$, $\{H_{2,k},\,k\in\Integers\}$, and $\{H_{3,k},\,k\in\Integers\}$ satisfy
\begin{equation*}
\eps_2^2 > \eps^2_{1} + \eps^2_3.
\end{equation*}
Then we have
\begin{equation}
\chi > -1 -\gamma + \log\frac{1}{\eps^2_2}.
\end{equation}
\end{corollary}

Corollaries~\ref{cor:direct} and \ref{cor:cooperative} demonstrate that direct communication from the transmitter to the receiver (i.e., turning the relay off) is optimal with respect to the fading number if the prediction error $\eps_2^2$ corresponding to the channel between the transmitter and receiver is not larger than the prediction error $\eps_3^2$ corresponding to the channel between the relay and the receiver. In contrast, cooperative communication is beneficial with respect to the fading number if the prediction error $\eps^2_2$ corresponding to the channel between the transmitter and receiver is larger than the sum of prediction errors $\eps_1^2+\eps_3^2$ corresponding to the channels from the transmitter to the relay and from the relay to the receiver. It is unknown whether cooperative communication is beneficial if the prediction errors satisfy $\eps_3^2 < \eps_2^2 \leq \eps_1^2+\eps_3^2$.

Denoting the fading number corresponding to the fading $\{H_{\ell,k},\,k\in\Integers\}$ by $\chi_{\ell}$, the lower bound \eqref{eq:thmlowerbound1} can be written as
\begin{IEEEeqnarray}{lCl}
\chi & \geq &  \max\Bigl\{\chi_2,\chi_3 - \log\Bigl(1+\exp\bigl(\chi_3-\chi_1\bigr)\Bigr)\Bigr\}\nonumber\\
& = & \max\Bigl\{\chi_2,\chi_1 - \log\Bigl(1+\exp\bigl(\chi_1-\chi_3\bigr)\Bigr)\Bigr\}. \label{eq:DF_1bit}
\end{IEEEeqnarray}
Note that if $\chi_1\geq \chi_3$, then
\begin{equation*}
\log\Bigl(1+\exp\bigl(\chi_3-\chi_1\bigr)\Bigr) \leq \log 2
\end{equation*}
and the difference between the lower bound \eqref{eq:DF_1bit} and the upper bound \eqref{eq:fadingnumberMISO} is at most one bit. This is summarized in the following corollary.
\begin{corollary}
\label{cor:1bit}
Let the fading processes $\{H_{1,k},\,k\in\Integers\}$ and $\{H_{3,k},\,k\in\Integers\}$ satisfy
\begin{equation*}
\eps_1^2 \leq \eps_3^2.
\end{equation*}
Then we have
\begin{equation}
0 \leq \max\biggl\{-1-\gamma+\log\frac{1}{\eps^2_{2}},-1-\gamma+\log\frac{1}{\eps^2_3}\biggr\} - \chi \leq \log 2.
\end{equation}
\end{corollary}

As observed above, for SNR values below 80dB, the capacity is approximately upper-bounded by
\begin{equation}
\label{eq:1bitboundslow}
C(\SNR,\rho) \lessapprox 3+\chi, \quad \SNR\leq80\textnormal{dB}
\end{equation}
so a gap of $\log 2 \approx 0.6931$ nats seems substantial. Nevertheless, for slowly-varying fading channels, the prediction errors $\eps^2_{\ell}$, $\ell=1,2,3$ are small and the fading number, which depends on $\eps^2_{\ell}$ via $-\log\eps^2_{\ell}$, is much larger than $\log 2$. For example, for mobile speeds of the order of 5 km/h, prediction errors $\eps_{\ell}^2$ of roughly $10^{-4}$ seem plausible; see, e.g., \cite[Sec.~II]{etkintse06}. In this case, the fading number is approximately
\begin{equation*}
\chi = -1 -\gamma + \log\frac{1}{\eps_{\ell}^2} \approx 7.6331\textnormal{ nats}
\end{equation*}
and the RHS of \eqref{eq:1bitboundslow} becomes $10.6331$ nats. Thus, for slowly-varying fading channels, a gap of one bit (or equivalently $\log 2$ nats) is reasonably small.

Corollary~\ref{cor:1bit} demonstrates that, when $\chi_1\geq\chi_3$, decode-and-forward achieves communication rates that are within one bit of the capacity of the relay channel. This is consistent with the \emph{Gaussian relay channel} where decode-and-forward also achieves rates that are within one bit of the capacity \cite[Th.~3.1]{avestimehrdiggavitse11}. Note that the difference between the lower bound \eqref{eq:DF_1bit} and the upper bound \eqref{eq:fadingnumberMISO} decreases as $(\chi_1-\chi_3)$ increases.

We conclude that if the fading coefficient of the channel between the transmitter and the relay can be predicted more accurately than the fading coefficient of the channel between the relay and the receiver, then the fading number of the fading relay channel is at most one bit smaller than the fading number of the TRC-MISO channel. If we view the fading number as an indication of the rates at which communication is power-inefficient, then this result demonstrates that the rates up to which the fading relay channel and the TRC-MISO channel operate in the power-efficient regime are within one bit. Note, however, that this does not imply that for both channels the power-inefficient regime starts at the same SNR. Indeed, in the following section we derive nonasymptotic upper and lower bounds on the capacity of the fading relay channel as well as on the capacity of the TRC-MISO channel. These bounds suggest that the capacity of the fading relay channel increases much more slowly with SNR than the capacity of the TRC-MISO channel.

\section{Nonasymptotic Bounds}
\label{sec:nonasymptotic}
To simplify the analysis, we assume throughout this section that the channel between the transmitter and the receiver is memoryless, i.e., we have
\begin{equation*}
F'_2(\lambda) = 1, \quad -\frac{1}{2} \leq\lambda\leq\frac{1}{2}
\end{equation*}
which yields $\eps_2^2=1$. For this case, a nonasymptotic upper bound on the capacity of the relay channel follows by letting the transmitter and the relay cooperate, by relaxing the power constraint to
\begin{equation}
\label{eq:relaxed_PP}
|X_k|^2 + |X_{r,k}|^2 \leq \const{A}^2(1+\rho^2), \quad k\in\Integers
\end{equation} 
and by extending \cite[Eq.~(16)] {kochlapidoth05_1} to the TRC-MISO channel:
\begin{IEEEeqnarray}{lCl}
\IEEEeqnarraymulticol{3}{l}{C(\SNR,\rho)}\nonumber\\
\quad & \leq & C_{\text{IID}}\bigl(\SNR(1+\rho^2)\bigr) + \log\biggl(1+\frac{1}{(1+\rho^2)\SNR}\biggr)\nonumber\\
& & {} - \int_{1/2}^{1/2} \log\biggl(F_3'(\lambda)+\frac{1}{(1+\rho^2)\SNR}\biggr) \d\lambda\label{eq:nonasymptotic_U}
\end{IEEEeqnarray}
where $C_{\text{IID}}(\SNR)$ denotes the capacity in the memoryless fading case, which can be upper-bounded by \cite[Eq.~(141)]{lapidothmoser03_3}
\begin{IEEEeqnarray}{lCl}
C_{\text{IID}}(\SNR) & \leq & \inf_{\substack{\alpha,\beta<0,\\\delta> 0}} \Biggl\{-1+\alpha\log\frac{\beta}{\delta} + \log\Gamma\biggl(\alpha,\frac{\delta}{\beta}\biggr) + \log\delta \nonumber\\
\IEEEeqnarraymulticol{3}{r}{ {}  - (1-\alpha) e^{\delta}\Ei{-\delta} + \frac{\SNR+1}{\beta} + \frac{\delta}{\beta} \Biggr\}.\quad\,\,\,} \label{eq:3(141)}
\end{IEEEeqnarray}
Here
\begin{equation*}
\Gamma(\nu,\xi) \triangleq \int_{\xi}^{\infty} t^{\nu-1} e^{-t} \d t, \quad \bigl(\nu>0,\,\xi\geq 0\bigr)
\end{equation*}
denotes the incomplete Gamma function and
\begin{equation*}
\Ei{-x} \triangleq -\int_{x}^{\infty} \frac{e^{-t}}{t}\d t, \quad x >0
\end{equation*}
denotes the exponential integral function. Clearly, the upper bound \eqref{eq:nonasymptotic_U} is also an upper bound on the capacity of the TRC-MISO channel. 

Alternatively, \eqref{eq:nonasymptotic_U} can be directly derived from \cite[Th.~4.2]{koch09} by setting $\nr=2$ in \cite[Eq.~(4.37)]{koch09} and computing
\begin{IEEEeqnarray*}{lCl}
\eps_{1,1}^2\biggl(\frac{1}{(1+\rho^2)\SNR}\biggr) & = & \log\biggl(1+\frac{1}{(1+\rho^2)\SNR}\biggr)
\end{IEEEeqnarray*}
and
\begin{IEEEeqnarray*}{lCl}
\IEEEeqnarraymulticol{3}{l}{\eps_{1,2}^2\biggl(\frac{1}{(1+\rho^2)\SNR}\biggr)} \nonumber\\
\quad & = & \int_{-1/2}^{1/2} \log\biggl(F'_3(\lambda) + \frac{1}{(1+\rho^2)\SNR}\biggr)\d\lambda;
\end{IEEEeqnarray*}
and by noting that, by Jensen's inequality, we have
\begin{IEEEeqnarray*}{lCl}
\eps_{1,2}^2\biggl(\frac{1}{(1+\rho^2)\SNR}\biggr) & \leq & \eps_{1,1}^2\biggl(\frac{1}{(1+\rho^2)\SNR}\biggr)
\end{IEEEeqnarray*}
implying that the maximum on the RHS of \cite[Eq.~(4.37)]{koch09} is achieved for $\bigl(|\hat{x}(1)|^2,|\hat{x}(2)|^2\bigr)=(0,1)$.

The next proposition presents a nonasymptotic lower bound on the capacity. It is based on a lower bound that was derived in \cite[Prop.~4.1]{koch09} for single-antenna point-to-point fading channels, combined with a decode-and-forward strategy. For point-to-point fading channels, this bound is tight at high SNR in the sense that it achieves the fading number. For the fading relay channel, it can be shown that this bound achieves the lower bound on the fading number given in Theorem~\ref{thm:lowerbound1}.

\begin{proposition}
\label{prop:nonasymptotic}
Let the fading process $\{H_{2,k},\,k\in\Integers\}$ be memoryless. A decode-and-forward strategy yields
\begin{IEEEeqnarray}{lCll}
C(\SNR,\rho)  & \geq & \sup_{0<\delta,\alpha,\delta_r<1} \min \bigl\{& R_{tr}(\SNR;\delta,\alpha),\nonumber\\
& & & \, R_{rr}(\SNR,\rho;\delta,\alpha,\delta_r)\bigr\}\IEEEeqnarraynumspace \label{eq:nonasymptotic_L}
\end{IEEEeqnarray}
where
\begin{IEEEeqnarray*}{rCl}
\IEEEeqnarraymulticol{3}{l}{R_{tr}(\SNR;\delta,\alpha)}\nonumber\\
& \triangleq & \log\biggl(\frac{\sigma^{2(1-\alpha)}}{\delta^{\alpha}\SNR^{\alpha}}\biggr) -\int_{-1/2}^{1/2} \log\biggl(F'_1(\lambda)+\frac{\sigma^{2(1-\alpha)}}{\delta^{2\alpha}\SNR^{\alpha}}\biggr)\d\lambda \nonumber\\
& & {} - \exp\Biggl(\frac{\sigma^{2(1-\alpha)} e}{\alpha\log\bigl(\frac{1}{\delta^2}\bigr)\delta^{\alpha}\SNR^{\alpha}}\Biggr)\Ei{-\frac{\sigma^{2(1-\alpha)} e}{\alpha\log\bigl(\frac{1}{\delta^2}\bigr)\delta^{\alpha}\SNR^{\alpha}}}
\end{IEEEeqnarray*}
and
\begin{IEEEeqnarray*}{lCl}
\IEEEeqnarraymulticol{3}{l}{R_{rr}(\SNR,\rho;\delta,\alpha,\delta_r)}\nonumber\\
\,\,\, & \triangleq & \log\Biggl(\frac{e^{-(\gamma+1)}\alpha\log\bigl(\frac{1}{\delta^2}\bigr)\delta^{\alpha}\SNR^{\alpha}\sigma^{2(\alpha-1)}+1}{\delta_r\rho^2\SNR}\Biggr) \nonumber\\
& & {} -\int_{-1/2}^{1/2} \log\biggl(F_3'(\lambda)+\frac{\sigma^{2(\alpha-1)}}{\delta_r^2\rho^2\SNR^{1-\alpha}}+\frac{1}{\delta_r^2\rho^2\SNR}\biggr)\d\lambda \nonumber\\
& & {} - \exp\Biggl(\frac{e^{-\gamma}\alpha\log\bigl(\frac{1}{\delta^2}\bigr)\delta^{\alpha}\SNR^{\alpha}\sigma^{2(\alpha-1)}+e}{\log\bigl(\frac{1}{\delta_r^2}\bigr)\delta_r\rho^2\SNR}\Biggr)\times\nonumber\\
& & \qquad {} \times\Ei{-\frac{e^{-\gamma}\alpha\log\bigl(\frac{1}{\delta^2}\bigr)\delta^{\alpha}\SNR^{\alpha}\sigma^{2(\alpha-1)}+e}{\log\bigl(\frac{1}{\delta_r^2}\bigr)\delta_r\rho^2\SNR}}.
\end{IEEEeqnarray*}
\end{proposition}
\begin{proof}
See Appendix~\ref{app:nonasymptotic}.
\end{proof}

A lower bound on the capacity of the TRC-MISO channel (but not necessarily the fading relay channel) follows by using a beam-selection strategy, where the transmitter transmits either from the first antenna (i.e., the transmitter) or from the second antenna (i.e., the relay). To compare the lower bound with the upper bound \eqref{eq:nonasymptotic_U}, we consider the relaxed power constraint \eqref{eq:relaxed_PP}. Note, however, that the main conclusions drawn from the nonasymptotic bounds would not change if we had considered the original power constraints \eqref{eq:peakpower1} and \eqref{eq:peakpower2} to lower-bound the capacity of the TRC-MISO channel. In fact, the fading number of the TRC-MISO channel is the same for both power constraints. The TRC-MISO capacity is lower-bounded by \cite[Prop.~4.1]{koch09}
\begin{IEEEeqnarray}{lCl}
\IEEEeqnarraymulticol{3}{l}{C_{\text{MISO}}(\SNR,\rho)}\nonumber\\
\quad & \geq & \sup_{0<\delta<1} \max\bigl\{R_{2}(\SNR,\rho;\delta), R_3(\SNR,\rho;\delta)\bigr\} \label{eq:nonasymptotic_TRC}
\end{IEEEeqnarray}
where
\begin{IEEEeqnarray*}{lCl}
\IEEEeqnarraymulticol{3}{l}{R_{\ell}(\SNR,\rho;\delta)}\nonumber\\
\quad & \triangleq & \log\biggl(\frac{1}{\delta\SNR(1+\rho^2)}\biggr) \nonumber\\
& & {} -\int_{-1/2}^{1/2} \log\biggl(F'_{\ell}(\lambda)+\frac{1}{\delta^{2}\SNR (1+\rho^2)}\biggr)\d\lambda
\nonumber\\
& & {} - \exp\Biggl(\frac{e}{\log\bigl(\frac{1}{\delta^2}\bigr)\delta\SNR(1+\rho^2)}\Biggr)\times\nonumber\\
& & \qquad\qquad {} \times\Ei{-\frac{e}{\log\bigl(\frac{1}{\delta^2}\bigr)\delta\SNR(1+\rho^2)}}
\end{IEEEeqnarray*}
for $\ell=2,3$. Note that beam-selection is optimal at high SNR in the sense that it achieves the fading number \cite{kochlapidoth05_1,kochlapidoth05_3,koch04}. 

The lower bounds \eqref{eq:nonasymptotic_L} and \eqref{eq:nonasymptotic_TRC} are tight at high SNR, but they are loose at low SNR. We therefore include the following lower bounds that are superior to \eqref{eq:nonasymptotic_L} and \eqref{eq:nonasymptotic_TRC} when the SNR is small.

A lower bound on the capacity of the TRC-MISO channel follows by using a beam-selection strategy and by lower-bounding the capacity of the resulting point-to-point channel by choosing quaternary phase-shift keying (QPSK) channel inputs \cite[Prop.~2.1 \& Eq.~(17)]{sethuramanhajeknarayanan05}:
\begin{IEEEeqnarray}{lCll}
C_{\text{MISO}}(\SNR) & \geq & \max\Bigl\{& I\bigl(X;(\bar{H}_{1}+\tilde{H}_{1})X+Z_{r}\bigm|\bar{H}_{1}\bigr), \nonumber\\
& & & \, I\bigl(X_{r};(\bar{H}_{3}+\tilde{H}_{3})X_r+Z\bigm|\bar{H}_{3}\bigr)\Bigr\}\IEEEeqnarraynumspace\label{eq:nonasymptotic_TRC_lowSNR}
\end{IEEEeqnarray}
where $\bar{H}_{1}$, $\tilde{H}_1$, $\bar{H}_3$, $\tilde{H}_{3}$, $X$, $X_r$, $Z_{r}$, and $Z$ are independent random variables with the following distributions:
\begin{itemize}
\item $X$ and $X_{r}$ are uniformly distributed over the QPSK constellation
\begin{IEEEeqnarray*}{l}
\Bigl\{\sqrt{(1+\rho^2)}\const{A},\ii\sqrt{(1+\rho^2)}\const{A},\\
 \quad {} -\sqrt{(1+\rho^2)}\const{A},-\ii\sqrt{(1+\rho^2)}\const{A}\Bigr\};
\end{IEEEeqnarray*}
\item $\bar{H}_{\ell}$, $\ell=1,3$ and $\tilde{H}_{\ell}$, $\ell=1,3$ are zero-mean, circularly-symmetric, complex Gaussian random variables of respective variances
\begin{equation*}
1-\eps^2_{\ell}\biggl(\frac{1}{(1+\rho^2)\SNR}\biggr) \quad \textnormal{and} \quad \eps^2_{\ell}\biggl(\frac{1}{(1+\rho^2)\SNR}\biggr)
\end{equation*}
where $\eps_{\ell}^2(\cdot)$ denotes the noisy prediction error \cite[Eq.~(11)]{lapidoth05}
\begin{equation*}
\eps_{\ell}(\xi) \triangleq \exp\Biggl(\int_{-1/2}^{1/2} \log(F'_{\ell}(\lambda)+\xi)\d\lambda\Biggr)-\xi;
\end{equation*}
\item $Z_r$ and $Z$ are zero-mean, circularly-symmetric, complex Gaussian random variables of variance $\sigma^2$.
\end{itemize}
The RHS of \eqref{eq:nonasymptotic_TRC_lowSNR} can be computed numerically.

The above lower bound can be extended to the relay channel by employing a decode-and-forward strategy and by choosing $\{X_k,\,k\in\Integers\}$ and $\{X_{r,k},\,k\in\Integers\}$ to be i.i.d. random variables, independent of each other and with
\begin{itemize}
\item $X_k$ being uniformly distributed over the QPSK constellation $\bigl\{\delta\const{A},\ii\delta\const{A},-\delta\const{A},-\ii\delta\const{A}\bigr\}$, for some $0<\delta\leq 1$;
\item $X_{r,k}$ being uniformly distributed over the QPSK constellation $\bigl\{\const{A},\ii\const{A},-\const{A},-\ii\const{A}\bigr\}$.
\end{itemize}
This yields
\begin{IEEEeqnarray}{lCl}
C(\SNR,\rho) & \geq & \sup_{0<\delta\leq 1} \min\Bigl\{I\bigl(X_1;(\bar{H}'_1+\tilde{H}'_1)X_1+Z_{r}\bigm|\bar{H}'_1\bigr),\nonumber\\
\IEEEeqnarraymulticol{3}{r}{I\bigl(X_{r,1};(\bar{H}'_3+\tilde{H}'_3)X_{r,1}+H_{2}X_1+Z\bigm|\bar{H}_3,X_1\bigr)\Bigr\}\quad\,\,}\label{eq:nonasymptotic_L_lowSNR}
\end{IEEEeqnarray}
where $\bar{H}_1$, $\tilde{H}'_1$, $\bar{H}'_3$, $\tilde{H}'_3$, $X_1$, $X_{r,1}$, $Z$, $Z_{r}$, and $H_{2}$ are independent random variables with the following distributions:
\begin{itemize}
\item $\bar{H}'_1$ and $\tilde{H}'_1$ are zero-mean, circularly-symmetric, complex Gaussian random variables of respective variances
\begin{equation*}
1-\eps_1^2\biggl(\frac{1}{\delta^{2}\SNR}\biggr) \quad \textnormal{and} \quad \eps_1^2\biggl(\frac{1}{\delta^2\SNR}\biggr);
\end{equation*}
\item $\bar{H}'_3$ and $\tilde{H}'_3$ are zero-mean, circularly-symmetric, complex Gaussian random variables of respective variances
\begin{equation*}
1-\eps_3^2\biggl(\frac{\delta^2}{\rho^2}+\frac{1}{\rho^2\SNR}\biggr) \quad \textnormal{and} \quad\eps_3^2\biggl(\frac{\delta^2}{\rho^2}+\frac{1}{\rho^2\SNR}\biggr);
\end{equation*} 
\item $H_2$ is a zero-mean, unit-variance, circularly-symmetric, complex Gaussian random variable;
\item $Z$ and $Z_r$ are as above.
\end{itemize}
The RHS of \eqref{eq:nonasymptotic_L_lowSNR} can be computed numerically.

\begin{figure}
\begin{flushleft}
  \psfrag{Upper Bound TRC-MISO Channel}{{\footnotesize Upper bound TRC-MISO channel}}
  \psfrag{Lower Bound Relay Channel}{{\footnotesize Lower bound fading relay channel}}
  \psfrag{Lower Bound TRC-MISO Channel}{{\footnotesize Lower bound TRC-MISO channel}}
   \psfrag{Upper Bound Direct Communication}{{\footnotesize Upper bound direct communication}}
  \psfrag{Low SNR Lower Bound}{{\footnotesize Low-SNR lower bound}}
   \psfrag{Fading Number MISO Channel}{{\footnotesize Fading number TRC-MISO channel}}
  \psfrag{Fading Number Fading Relay Channel}{{\footnotesize Fading number relay channel}}
  \psfrag{e1}[b][b]{{$\eps_1^2=10^{-4}$, $\eps_2^2=1$, $\eps_3^2=10^{-2}$}}
  \psfrag{e2}[b][b]{{$\eps_1^2=10^{-2}$, $\eps_2^2=1$, $\eps_3^2=10^{-2}$}}
   \psfrag{SNR [dB]}[ct][cc]{SNR [dB]}
  \psfrag{Capacity [nats/channel use]}[cb][cb]{Capacity [nats/channel use]}
  \psfrag{0}[rt][rt]{\scriptsize$0$}
  \psfrag{1}[r][r]{\scriptsize$1$}
  \psfrag{2}[r][r]{\scriptsize$2$}
  \psfrag{3}[r][r]{\scriptsize$3$}
  \psfrag{4}[r][r]{\scriptsize$4$}
  \psfrag{5}[r][r]{\scriptsize$5$}
  \psfrag{6}[r][r]{\scriptsize$6$}
  \psfrag{7}[r][r]{\scriptsize$7$}
  \psfrag{-10}[t][t]{\scriptsize$-10$}
 \psfrag{10}[t][t]{\scriptsize$10$}
 \psfrag{20}[t][t]{\scriptsize$20$}
 \psfrag{30}[t][t]{\scriptsize$30$}
 \psfrag{40}[t][t]{\scriptsize$40$}
 \psfrag{50}[t][t]{\scriptsize$50$}
 \psfrag{60}[t][t]{\scriptsize$60$}
 \psfrag{70}[t][t]{\scriptsize$70$}
 \psfrag{80}[t][t]{\scriptsize$80$}
 \psfrag{90}[t][t]{\scriptsize$90$}
  \begin{center}
    \includegraphics[width=0.5\textwidth]{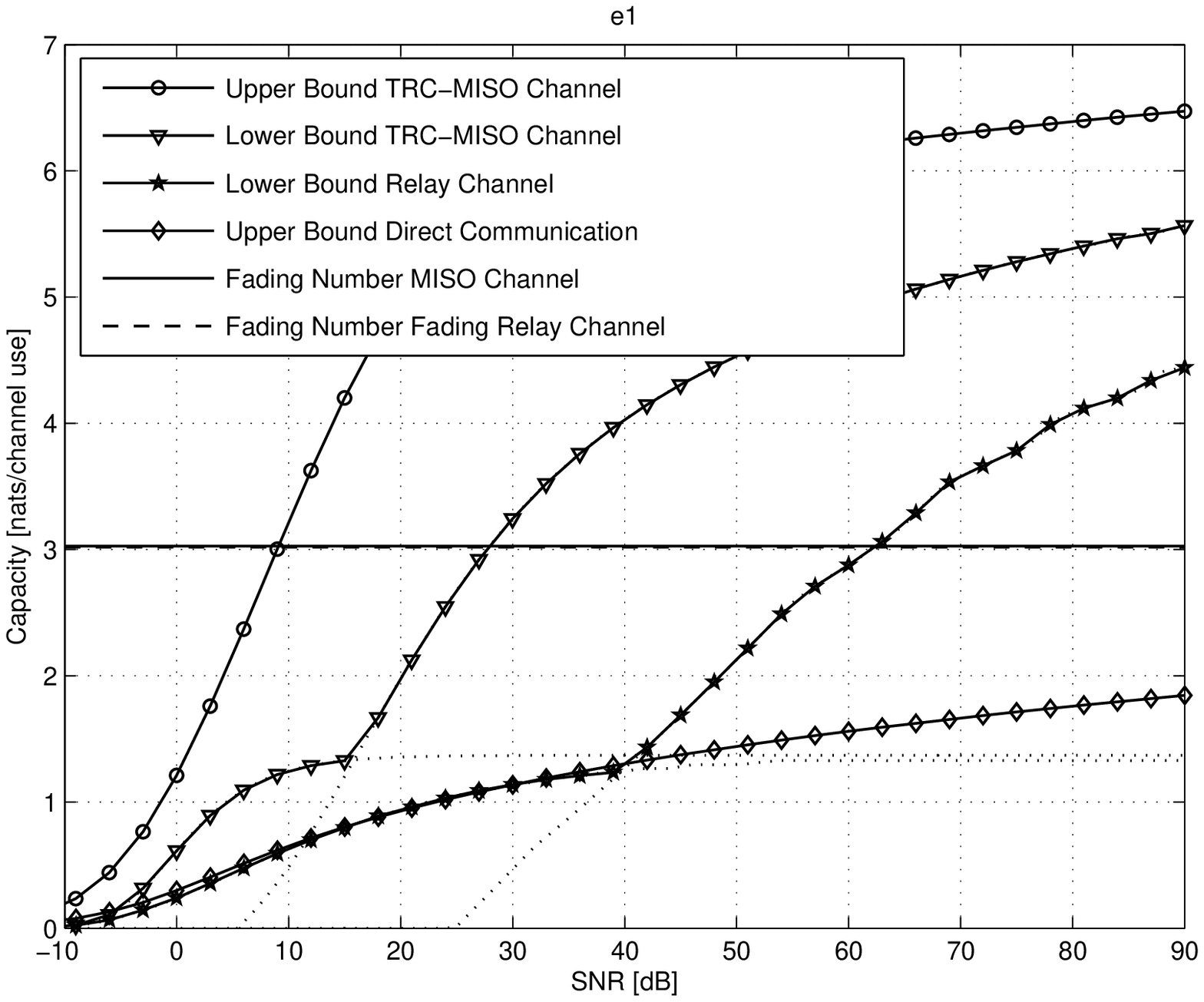}
  \end{center}
  \begin{center}
    \includegraphics[width=0.5\textwidth]{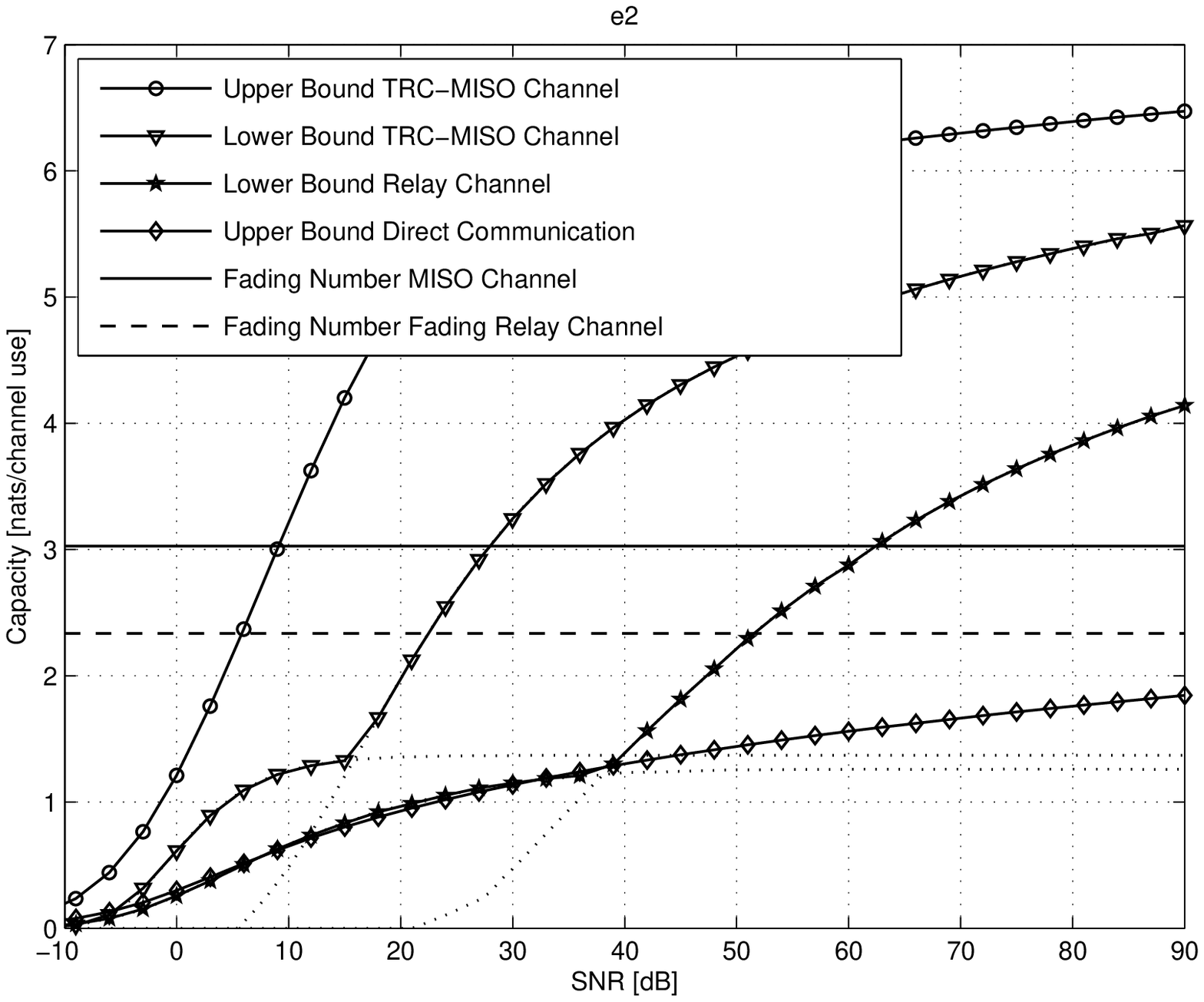}
  \end{center}
  \caption{Upper bound on the capacity of the TRC-MISO channel \eqref{eq:nonasymptotic_U}; lower bounds on the capacity of the TRC-MISO channel (maximum of \eqref{eq:nonasymptotic_TRC} and \eqref{eq:nonasymptotic_TRC_lowSNR}) and on the capacity of the fading relay channel (maximum of \eqref{eq:nonasymptotic_L} and \eqref{eq:nonasymptotic_L_lowSNR}); upper bound on the capacity for direct communication \eqref{eq:3(141)}; the fading number of the TRC-MISO channel and the lower bound \eqref{eq:thmlowerbound1} on the fading number of the fading relay channel. The prediction errors are $\eps_1^2=10^{-4}$, $\eps_2^2=1$, $\eps_3^2=10^{-2}$ (top subfigure) and $\eps_1^2=10^{-2}$, $\eps_2^2=1$, $\eps_3^2=10^{-2}$ (bottom subfigure). In both cases we assume $\rho=1$ and $\sigma=1$.}
  \label{fig:nonasymp}
  \end{flushleft}
\end{figure}

We evaluate the upper bound \eqref{eq:nonasymptotic_U} and the lower bounds \eqref{eq:nonasymptotic_L}--\eqref{eq:nonasymptotic_L_lowSNR} for spectral distribution functions of the form
\begin{equation}
F_{\ell}'(\lambda) = \left\{\begin{array}{ll} \Upsilon_{\ell}, \quad & |\lambda|\leq\Theta_{\ell} \\ \Lambda_{\ell}, \quad & \Theta_{\ell}< |\lambda|\leq\frac{1}{2}, \end{array}\right. \quad \ell=1,3
\end{equation}
where $\Upsilon_{\ell}>0$, $\Lambda_{\ell}>0$, and $0<\Theta_{\ell}<1/2$ satisfy
\begin{equation*}
\int_{-1/2}^{1/2} F'_{\ell}(\lambda)\d\lambda = 2\Upsilon_{\ell}\Theta_{\ell}+ \bigl(1-2\Theta_{\ell}\bigr)\Lambda_{\ell} =1, \quad \ell=1,3.
\end{equation*}
Recall that $F_2'(\lambda)=1$, $1/2\leq\lambda\leq 1/2$. We consider two scenarios: in the first scenario, the fading between the transmitter and the relay has a prediction error of $10^{-4}$, whereas the fading between the relay and the receiver has a prediction error of $10^{-2}$. This implies that the fading number of the relay channel is roughly the same as the fading number of the TRC-MISO channel.  In the second scenario, both the fading between the transmitter and the relay and the fading between the relay and the receiver have a prediction error of $10^{-2}$. In this case, the lower bound on the fading number of the relay channel \eqref{eq:thmlowerbound1} is $\log 2$ nats smaller than the fading number of the TRC-MISO channel.

Figure~\ref{fig:nonasymp} shows the upper bound on the capacity of the fading relay channel and of the TRC-MISO channel \eqref{eq:nonasymptotic_U}, the lower bounds on the capacity of the fading relay channel (maximum of \eqref{eq:nonasymptotic_L} and \eqref{eq:nonasymptotic_L_lowSNR}) and of the TRC-MISO channel (maximum of \eqref{eq:nonasymptotic_TRC} and \eqref{eq:nonasymptotic_TRC_lowSNR}), together with the corresponding fading numbers for the above two scenarios. In particular, the top subfigure in Figure~\ref{fig:nonasymp} shows the bounds \eqref{eq:nonasymptotic_U}, \eqref{eq:nonasymptotic_L}--\eqref{eq:nonasymptotic_L_lowSNR} for
\begin{equation*}
\begin{array}{lcl} \Upsilon_{1} & \approx & 5.76034 \\
\Lambda_{1} & = & 10^{-5} \\
\Theta_{1} & \approx & 0.08679
\end{array}
\qquad \text{and} \qquad
\begin{array}{lcl} \Upsilon_{3} & \approx & 10.99684 \\
\Lambda_{3} & = & 0.005 \\
\Theta_{3} & \approx & 0.04503
\end{array}
\end{equation*}
resulting in $\eps_1^2=10^{-4}$ and $\eps_3^{2}=10^{-2}$. In this case, the fading number is upper-bounded by \eqref{eq:thmUB}, which yields
\begin{equation*}
\chi \leq - 1- \gamma+\log\frac{1}{\eps_3^2} \approx 3.0280 \textnormal{ nats}
\end{equation*}
and is equal to the fading number of the TRC-MISO channel. The fading number of the relay channel is lower-bounded by \eqref{eq:thmlowerbound1}, which is
\begin{equation*}
\chi \geq -1-\gamma+\log\frac{1}{\eps_3^2} - \log\biggl(1+\frac{\eps^2_1}{\eps^2_3}\biggr) \approx 3.0180 \textnormal{ nats.}
\end{equation*}
The bottom subfigure in Figure~\ref{fig:nonasymp} shows the bounds  \eqref{eq:nonasymptotic_U}, \mbox{\eqref{eq:nonasymptotic_L}--\eqref{eq:nonasymptotic_L_lowSNR}} for
\begin{equation*}
\begin{array}{lclcl}  \Upsilon_{1} & = & \Upsilon_{3} & \approx & 10.99684 \\
\Lambda_{1} & = & \Lambda_{3} & = & 0.005 \\
\Theta_{1} & = & \Theta_{3} & \approx & 0.04503
\end{array}
\end{equation*}
resulting in $\eps_1^2=\eps_3^{2}=10^{-2}$. As in the above example, for the relay channel this yields
\begin{equation*}
\chi \leq - 1- \gamma+\log\frac{1}{\eps_3^2} \approx 3.0280 \textnormal{ nats}
\end{equation*}
which is equal to the fading number of the TRC-MISO fading channel. The lower bound \eqref{eq:thmlowerbound1} becomes
\begin{equation*}
\chi \geq -1-\gamma+\log\frac{1}{\eps_3^2} - \log 2 \approx 2.3348 \textnormal{ nats.}
\end{equation*}
In both examples, we assume that $\rho=1$ and $\sigma=1$.

To compare the performance of cooperative communication with that of direct communication, we also show an upper bound on the capacity when the relay is switched off. Since in this section we assume that the channel between the transmitter and receiver is memoryless, it follows that, when the relay is switched off, the capacity for both examples is upper-bounded by \eqref{eq:3(141)}, which is
\begin{IEEEeqnarray*}{lCll}
\IEEEeqnarraymulticol{4}{l}{C(\SNR)}\nonumber\\
\quad & \leq & \inf_{\substack{\alpha,\beta<0,\\\delta> 0}} \Biggl\{& -1+\alpha\log\frac{\beta}{\delta} + \log\Gamma\biggl(\alpha,\frac{\delta}{\beta}\biggr) + \log\delta \nonumber\\
& & & \, {}  - (1-\alpha) e^{\delta}\Ei{-\delta} + \frac{\SNR+1}{\beta} + \frac{\delta}{\beta} \Biggr\}. \IEEEeqnarraynumspace 
\end{IEEEeqnarray*}

Observe that the lower bound for the fading relay channel \eqref{eq:nonasymptotic_L} increases much more slowly with SNR than the lower bound for the TRC-MISO channel \eqref{eq:nonasymptotic_TRC}, even in the first example where the fading numbers of both channels are almost identical. Since these lower bounds are tight at high SNR (in the sense that they achieve the fading number of the TRC-MISO channel and the lower bound on the fading number of the relay channel, respectively) we suspect that the same is also true for the capacities of both channels at high SNR. Thus, even though the capacities of the fading relay channel and the TRC-MISO channel have similar asymptotic behaviors in the limit as the SNR tends to infinity, they may differ substantially at finite SNR.

Further observe that for SNR values below 40 dB the upper bound corresponding to direct communication \eqref{eq:3(141)} is comparable to the lower bound \eqref{eq:nonasymptotic_L_lowSNR} corresponding to cooperative communication, whereas for SNR values above 40dB the upper bound \eqref{eq:3(141)} is significantly smaller than the lower bound \eqref{eq:nonasymptotic_L} achievable with cooperation. We thus conclude that cooperation can provide significant capacity gains over direct communication for intermediate to large SNR values. Furthermore, since \eqref{eq:3(141)} does not seem to be tight for low to intermediate SNR (cf.~\cite[Fig.~1]{lapidothmoser03_3}), we expect that cooperation is also beneficial at SNR values below 40 dB.

Note that, for the above spectral distribution functions, the fading processes $\{H_{1,k},\,k\in\Integers\}$ and $\{H_{3,k},\,k\in\Integers\}$ are \emph{nonephermal} \cite[Def.~2.1]{sethuramanwanghajeklapidoth09} in the sense that
\begin{equation*}
\int_{-1/2}^{1/2} F'^2_{\ell}(\lambda) \d\lambda > 2, \quad \ell=1,3.
\end{equation*}
In this case i.i.d.\ inputs and QPSK as well as beam-selection achieve the low-SNR asymptotic capacity \cite[Secs.~II-A3 \& II-B4]{sethuramanwanghajeklapidoth09}. Thus, for the above spectral distribution functions, the lower bound \eqref{eq:nonasymptotic_TRC_lowSNR} is tight at low SNR. Further note that, while the high-SNR results presented in Section~\ref{sec:results} continue to hold if the peak-power constraints 
\eqref{eq:peakpower1} and \eqref{eq:peakpower2} are replaced by average-power constraints, this is not necessarily true for the nonasymptotic bounds presented in this section. In fact, for a peak-power constraint, the low-SNR asymptotic capacity behaves like $\SNR^2$, i.e., we have \cite[Cor.~2.1]{sethuramanwanghajeklapidoth09}
\begin{equation}
\lim_{\SNR\downarrow 0} \frac{C(\SNR)}{\SNR^2} = \kappa
\end{equation}
(where $\kappa$ is a function of the fading process's spectral distribution function), whereas for an average-power constraint, it behaves like $\SNR$, i.e., we have \cite[Th.~5.1.1]{lapidothshamai02}, \cite[Th.~1]{verdu02}
\begin{equation}
\label{eq:CUC_Verdu}
\lim_{\SNR\downarrow 0} \frac{C(\SNR)}{\SNR} = 1.
\end{equation}
Furthermore, in the average-power-limited case, QPSK channel inputs are highly suboptimal: \eqref{eq:CUC_Verdu} can be achieved only by input distributions that are \emph{flash signaling} \cite[Th.~7]{verdu02}. 

\section{Proof of Theorem~\ref{thm:upperbound}}
\label{sec:upperbound}
Theorem~\ref{thm:upperbound} follows from Fano's inequality \cite[Th.~2.11.1]{coverthomas91} and from the following upper bound on $\frac{1}{n}I\bigl(X_1^n;Y_1^n\bigr)$
\begin{IEEEeqnarray}{lCl}
C(\SNR,\rho) & \leq & \varlimsup_{n\to\infty}\sup \frac{1}{n} I\bigl(X_1^n;Y_1^n\bigr) \nonumber\\
& \leq & \min\biggl\{\varlimsup_{n\to\infty} \sup \frac{1}{n} I\bigl(X_1^n;Y_{r,1}^n,Y_1^n\bigr),\nonumber\\
& & {} \quad\qquad \varlimsup_{n\to\infty}\frac{1}{n} \sup I\bigl(X_1^n,X_{r,1}^n;Y_1^n\bigr)\biggr\} \IEEEeqnarraynumspace\label{eq:UB1}
\end{IEEEeqnarray}
where the suprema are over all joint distributions of $\bigl(X_1^n,X_{r,1}^n\bigr)$ satisfying the power constraints \eqref{eq:peakpower1} and \eqref{eq:peakpower2}. Here the second step follows by upper-bounding
\begin{equation*}
I\bigl(X_1^n;Y_1^n\bigr)\leq  I\bigl(X_1^n;Y_{r,1}^n,Y_1^n\bigr)
\end{equation*}
and
\begin{equation*}
I\bigl(X_1^n;Y_1^n\bigr)\leq I\bigl(X_1^n,X_{r,1}^n;Y_1^n\bigr)
\end{equation*}
which in turn follows because $I(A;B)\leq I(A;B,C)$ for every random variables $A$, $B$, and $C$.

The first term on the RHS of \eqref{eq:UB1} is upper-bounded by the capacity of a single-input multiple-output (SIMO) fading channel with peak-power $\const{A}^2$, and the second term on the RHS of \eqref{eq:UB1} is upper-bounded by the capacity of a MISO fading channel with peak-power $\const{A}^2+\const{A}_r^2$, which by \eqref{eq:alpha} is equal to $\const{A}^2\,(1+\rho^2)$. 

Indeed, since $X_1^n$ is independent of $H_{3,1}^n$, we can upper-bound the first term on the RHS of \eqref{eq:UB1} by
\begin{IEEEeqnarray}{lCl}
\frac{1}{n} I\bigl(X_1^n;Y_{r,1}^n,Y_1^n\bigr) & \leq & \frac{1}{n}  I\bigl(X_1^n;Y_{r,1}^n,Y_1^n \bigm| H_{3,1}^n\bigr). \label{eq:whatever}
\end{IEEEeqnarray}
Using the chain rule for mutual information \cite[Th.~2.5.2]{coverthomas91}, we obtain
\begin{IEEEeqnarray}{lCl}
\IEEEeqnarraymulticol{3}{l}{I\bigl(X_1^n;Y_{r,1}^n,Y_1^n \bigm| H_{3,1}^n\bigr)}\nonumber\\
\quad & = & \frac{1}{n} \sum_{k=1}^n I\bigl(X_1^n;Y_{r,k},Y_k\bigm| H_{3,1}^n, Y_{r,1}^{k-1},Y_1^{k-1}\bigr). \nonumber\\
& = & \frac{1}{n} \sum_{k=1}^n I\bigl(X_1^k;Y_{r,k},Y_k\bigm| H_{3,1}^k=0,Y_{r,1}^{k-1},Y_1^{k-1}\bigr) \nonumber\\
& \leq & \frac{1}{n} \sum_{k=1}^n I\bigl(X_1^k,Y_{r,1}^{k-1},Y_1^{k-1};Y_{r,k},Y_k\bigm| H_{3,1}^k=0\bigr) \IEEEeqnarraynumspace
\end{IEEEeqnarray}
where the second equality follows because $X_{r,1}^k$ is a function of $Y_{r,1}^{k-1}$, so $\bigl(H_{3,1}^n,X_{r,1}^k\bigr)$ is known and we can therefore subtract $H_{3,\ell} X_{r,\ell}$ from $Y_{\ell}$, $\ell=1,\ldots,k$, resulting in the same mutual information as if we would set $H_{3,1}^k=0$; and the subsequent inequality follows because $I(A;B|C)\leq I(A,C;B)$ for any random variables $A$, $B$, and $C$.

We next note that, conditioned on \[\bigl(X_k,H_{3,1}^k=0,H_{1,1}^{k-1},H_{2,1}^{k-1}\bigr)\] the pair $\bigl(Y_{r,k},Y_k\bigr)$ is independent of $\bigl(X_1^{k-1},Y_{r,1}^{k-1},Y_1^{k-1}\bigr)$, so adding the observations $\bigl(H_{1,1}^{k-1},H_{2,1}^{k-1}\bigr)$ yields
\begin{IEEEeqnarray}{lCl}
 \IEEEeqnarraymulticol{3}{l}{\frac{1}{n} \sum_{k=1}^n I\bigl(X_1^k,Y_{r,1}^{k-1},Y_1^{k-1};Y_{r,k},Y_k\bigm| H_{3,1}^k=0\bigr)}\nonumber\\
 \quad & \leq & \frac{1}{n} \sum_{k=1}^n I\bigl(X_k,H_{1,1}^{k-1},H_{2,1}^{k-1}; Y_{r,k},Y_k\bigm| H_{3,1}^k=0\bigr).\IEEEeqnarraynumspace
 \end{IEEEeqnarray}
The chain rule for mutual information gives
 \begin{IEEEeqnarray}{lCl}
 \IEEEeqnarraymulticol{3}{l}{\frac{1}{n} \sum_{k=1}^n I\bigl(X_k,H_{1,1}^{k-1},H_{2,1}^{k-1}; Y_{r,k},Y_k\bigm| H_{3,1}^k=0\bigr)} \nonumber\\
\quad & \leq &  \frac{1}{n} \sum_{k=1}^n \sup I\bigl(X_k;Y_{r,k},Y_k\bigm| H_{3,k}=0\bigr) \nonumber\\
& & {}  + \frac{1}{n} \sum_{k=1}^n I\bigl(H_{1,1}^{k-1},H_{2,1}^{k-1}; Y_{r,k},Y_k\bigm| X_k, H_{3,k}=0\bigr)\nonumber\\
& \leq & \sup I\bigl(X_1;Y_{r,1},Y_1\bigm| H_{3,1}=0\bigr) \nonumber\\
& & {} + \frac{1}{n} \sum_{k=1}^n I\bigl(H_{1,1}^{k-1},H_{2,1}^{k-1}; Y_{r,k},Y_k\bigm| X_k, H_{3,k}=0\bigr) \IEEEeqnarraynumspace \label{eq:UB2}
\end{IEEEeqnarray}
where the first supremum is over all input distributions of $X_k$ satisfying \eqref{eq:peakpower1}, and the second supremum is over all input distributions of $X_1$ satisfying \eqref{eq:peakpower1}. The first inequality in \eqref{eq:UB2} follows by upper-bounding each summand in the first sum by its supremum; and the second inequality follows from the stationarity of the channel, which is implies that $\sup I\bigl(X_k;Y_{r,k},Y_k\bigm| H_{3,k}=0\bigr)$ does not depend on $k$.

The first term on the RHS of \eqref{eq:UB2} is the capacity of the memoryless SIMO fading channel, which is given by \cite[Cor.~4.32]{lapidothmoser03_3}
\begin{IEEEeqnarray}{lCl}
\IEEEeqnarraymulticol{3}{l}{\sup  I\bigl(X_1;Y_{r,1},Y_1\bigm| H_{3,1}=0\bigr)\qquad} \nonumber\\
\IEEEeqnarraymulticol{3}{r}{{} = \log\log\SNR - 2\gamma + o(1)} \label{eq:UBmemoryless}
\end{IEEEeqnarray}
where $o(1)$ tends to zero as $\SNR\to\infty$. The second term on the RHS of \eqref{eq:UB2} can be upper-bounded by
\begin{IEEEeqnarray}{lCl}
\IEEEeqnarraymulticol{3}{l}{\frac{1}{n} \sum_{k=1}^n I\bigl(H_{1,1}^{k-1},H_{2,1}^{k-1}; Y_{r,k},Y_k\bigm| X_k, H_{3,k}=0\bigr)} \nonumber\\
\quad & \leq & \frac{1}{n} \sum_{k=1}^n I\bigl(H_{1,1}^{k-1},H_{2,1}^{k-1}; H_{1,k},H_{2,k}\bigr) \nonumber\\
& = & \frac{1}{n} \sum_{k=1}^n \biggl[I\bigl(H_{1,1}^{k-1}; H_{1,k}\bigr)+I\bigl(H_{2,1}^{k-1}; H_{2,k}\bigr)\biggr] \label{eq:UB3}\IEEEeqnarraynumspace
\end{IEEEeqnarray}
where the first step follows by the Data Processing Inequality \cite[Th.~2.8.1]{coverthomas91}; the second step follows because the processes $\{H_{1,k},\,k\in\Integers\}$ and $\{H_{2,k},\,k\in\Integers\}$ are independent.
By Ces\`aro's mean \cite[Th.~4.2.3]{coverthomas91}, it follows that
\begin{IEEEeqnarray}{lCl}
\IEEEeqnarraymulticol{3}{l}{\lim_{n\to\infty}\frac{1}{n} \sum_{k=1}^n \biggl[I\bigl(H_{1,1}^{k-1}; H_{1,k}\bigr)+I\bigl(H_{2,1}^{k-1}; H_{2,k}\bigr)\biggr]}\nonumber\\
\quad & = & \lim_{k\to\infty} I\bigl(H_{1,1}^{k-1}; H_{1,k}\bigr)+\lim_{k\to\infty} I\bigl(H_{2,1}^{k-1}; H_{2,k}\bigr) \nonumber\\
& = & \log\frac{1}{\eps_1^2} + \log\frac{1}{\eps_2^2} \label{eq:UB4}
\end{IEEEeqnarray}
where the second step follows because $\{H_{1,k},\,k\in\Integers\}$ and $\{H_{2,k},\,k\in\Integers\}$ are unit-variance Gaussian processes whose conditional variances, conditioned on the past $(k-1)$ fading coefficients, tend to $\eps_1^2$ and $\eps_2^2$ as $k$ tends to infinity \cite[Lemmas~5.7(b) \& 5.10(c)]{wienermasani57}. Combining \eqref{eq:whatever}--\eqref{eq:UB4}, we obtain
\begin{IEEEeqnarray}{lCl}
\IEEEeqnarraymulticol{3}{l}{\varlimsup_{n\to\infty}\sup \frac{1}{n} I\bigl(X_1^n;Y_{r,1}^n,Y_1^n\bigr)}\nonumber\\
\quad & \leq & \log\log\SNR-2\gamma+\log\frac{1}{\eps_1^2} + \log\frac{1}{\eps_2^2} + o(1). \label{eq:UBfirst}
\end{IEEEeqnarray}

To evaluate the second term on the RHS of \eqref{eq:UB1}, we note that by \eqref{eq:peakpower1}--$\eqref{eq:alpha}$ the channel inputs $X_k$ and $X_{r,k}$ satisfy\begin{equation}
|X_k|^2 + |X_{r,k}|^2 \leq \const{A}^2\,(1+\rho^2), \quad k\in\Integers\label{eq:UBpowerconstraint}
\end{equation}
with probability one. By maximizing over all joint distributions of $\bigl(X_1^n,X_{r,1}^n\bigr)$ satisfying \eqref{eq:UBpowerconstraint} (rather than \eqref{eq:peakpower1} and \eqref{eq:peakpower2}), it follows that
\[\varlimsup_{n\to\infty} \sup \frac{1}{n} I\bigl(X_1^n,X_{r,1}^n;Y_1^n\bigr)\]
is upper-bounded by the capacity of a MISO fading channel with fading processes $\{H_{2,k},\,k\in\Integers\}$ and $\{H_{3,k},\,k\in\Integers\}$ and with peak-power constraint $\const{A}^2(1+\rho^2)$. Consequently, we have \cite[Cor.~8]{kochlapidoth05_3}, \cite[Cor.~8]{kochlapidoth05_1}, \cite[Cor.~5.6]{koch04}
\begin{IEEEeqnarray}{lCl}
\IEEEeqnarraymulticol{3}{l}{\varlimsup_{n\to\infty} \sup \frac{1}{n} I\bigl(X_1^n,X_{r,1}^n;Y_1^n\bigr)} \nonumber\\
\quad & \leq & \log\log\Bigl(\SNR\,\bigl(1+\rho^2\bigr)\Bigr)-1-\gamma \nonumber\\
& & {} + \max\biggl\{\log\frac{1}{\eps^2_2},\log\frac{1}{\eps^2_3}\biggr\} + o(1). \label{eq:UBsecond}
\end{IEEEeqnarray}
Combining \eqref{eq:UBfirst} and \eqref{eq:UBsecond} with \eqref{eq:UB1}, we obtain
\begin{IEEEeqnarray}{lCll}
\IEEEeqnarraymulticol{4}{l}{C(\SNR,\rho)} \nonumber\\
\, & \leq & \min\Biggl\{& \log\log\SNR-2\gamma+\log\frac{1}{\eps_1^2}+\log\frac{1}{\eps^2_2},\nonumber\\
& & & \, \log\log\Bigl(\SNR\,\bigl(1+\rho^2\bigr)\Bigr) -1-\gamma\nonumber\\
& & &  \quad\qquad\qquad {} + \max\biggl\{\log\frac{1}{\eps_2^2},\log\frac{1}{\eps_3^2}\biggr\}\Biggr\}+o(1).\IEEEeqnarraynumspace
\end{IEEEeqnarray}
Computing the difference $\bigl\{C(\SNR,\rho)-\log\log\SNR\bigr\}$ in the limit as the SNR tends to infinity and using
\begin{equation}
\lim_{\SNR\to\infty} \bigl\{\log\log\Bigl(\SNR\,\bigl(1+\rho^2\bigr)\Bigr)-\log\log\SNR\bigr\} = 0
\end{equation}
for every fixed $\rho>0$, it follows that
\begin{IEEEeqnarray}{lCll}
\chi & \leq &  \min\Biggl\{& -2\gamma+\log\frac{1}{\eps_1^2}+\log\frac{1}{\eps_2^2},\nonumber\\
& & & \, \max\biggl\{-1-\gamma+\log\frac{1}{\eps_2^2},-1-\gamma+\log\frac{1}{\eps^2_3}\biggr\}\Biggr\}.\IEEEeqnarraynumspace
\end{IEEEeqnarray}
This proves Theorem~\ref{thm:upperbound}.

\section{Proof of Theorem~\ref{thm:lowerbound1}}
\label{sec:lowerbound}
To prove Theorem~\ref{thm:lowerbound1}, note that the first term in \eqref{eq:thmlowerbound1} is the fading number of the channel between the transmitter and receiver \cite[Cor.~4.42]{lapidothmoser03_3} and is achieved by turning the relay off. It thus remains to derive the second term, which follows from the following proposition.

\begin{proposition}[Decode-and-forward]
\label{prop:decodeforward}
Consider the fading relay channel described in Section~\ref{sec:channel}. Then the rate
\begin{IEEEeqnarray}{lCll}
R & = & \sup\lim_{n\to\infty}\frac{1}{n} \min\Bigl\{& I\bigl(X_1^n;Y_{r,1}^n\bigm| X_{r,1}^n\bigr), \nonumber\\
& & & \quad I\bigl(X_1^n,X_{r,1}^n;Y_1^n\bigr)\Bigr\} \label{eq:propDF}
\end{IEEEeqnarray}
is achievable. The supremum is over all i.i.d.\ processes $\{(X_k,X_{r,k}),\,k\in\Integers\}$
satisfying \eqref{eq:peakpower1} and \eqref{eq:peakpower2}.
\end{proposition}
\begin{proof}
See Appendix~\ref{app:decodeforward}.
\end{proof}
Proposition~\ref{prop:decodeforward} extends the decode-and-forward scheme proposed in \cite[Th.~1]{coverelgamal79} to channels with memory.

Theorem~\ref{thm:lowerbound1} follows from Proposition~\ref{prop:decodeforward} upon choosing $\{X_k,\,k\in\Integers\}$ and $\{X_{r,k},\,k\in\Integers\}$ to be i.i.d., circularly-symmetric, complex random variables, independent of each other and with
\begin{IEEEeqnarray}{rCll}
\log |X_k|^2 & \sim & \Uniform{\bigl[\log\log\const{A}^2,\log\const{A}^{2\alpha}\bigr]}, \quad & k\in\Integers\label{eq:DFindep1}\\
\log |X_{r,k}|^2 & \sim & \Uniform{\bigl[\log\const{A}_r^{2\beta},\log\const{A}_r^2\bigr]}, \quad & k\in\Integers\IEEEeqnarraynumspace\label{eq:DFindep2}
\end{IEEEeqnarray}
for some $0<\alpha<\beta<1$. Here $\Uniform{\set{S}}$ denotes the uniform distribution over the set $\set{S}$.

Before we prove Theorem~\ref{thm:lowerbound1}, we pause for intuition. Recall that if the channel between the transmitter and the receiver has a larger fading number than the channel between the relay and the receiver, then it is optimal to turn the relay off. This happens if $\chi_2\geq\chi_3$, and we therefore focus on the case where $\chi_2<\chi_3$. It follows from \eqref{eq:fadingnumberMISO} that in this case $\chi\leq\chi_3$. Since every signal sent from the transmitter to the relay interferes at the receiver, there is a trade-off between achieving high data rates from the transmitter to the relay (requiring a large transmit power) and minimizing the interference at the receiver (requiring a low transmit power). In order to attain a fading number that is close to the upper bound $\chi_3$, we choose $P_{X,X_r}(\cdot)$ such that $X_k/X_{r,k}$ vanishes as $\const{A}^2$ tends to infinity, thereby minimizing the interference at the receiver.

The input distribution \eqref{eq:DFindep1} and \eqref{eq:DFindep2} 
trades rates from the transmitter to the relay against rates from the relay to the receiver by using the parameters $\alpha$ and $\beta$. For instance, increasing $\alpha$ allows for larger rates between the transmitter and the relay, but requires a larger $\beta$ (since we have the condition $\beta>\alpha$) that decreases the rates achievable between the relay and the receiver.

\subsection{Lower bound on $\lim_{n\to\infty}\frac{1}{n}I\bigl(X_1^n;Y_{r,1}^n\bigm|X_{r,1}^n\bigr)$}
We lower-bound the first term on the RHS of \eqref{eq:propDF} via:
\begin{IEEEeqnarray}{lCl}
I\bigl(X_1^n;Y_{r,1}^n\bigm|X_{r,1}^n\bigr) & = & I\bigl(X_1^n;Y_{r,1}^n\bigr) \nonumber\\
& = & \sum_{k=1}^n I\bigl(X_k;Y_{r,1}^n\bigm| X_1^{k-1}\bigr) \nonumber\\
& \geq & \sum_{k=\kappa+1}^n I\bigl(X_k;Y_{r,1}^n\bigm| X_1^{k-1}\bigr) \label{eq:GK_1}
\end{IEEEeqnarray}
for some arbitrary $0\leq\kappa<n$. The first step in \eqref{eq:GK_1} follows because $X_1^n$ and $X_{r,1}^n$ are independent, and because $X_1^n$ and $X_{r,1}^n$ are also independent when conditioned on $Y_{r,1}^n$; the second step follows from the chain rule for mutual information; and the third step follows from the nonnegativity of mutual information. Using that $\{X_k,\,k\in\Integers\}$ is i.i.d.\ and that reducing observations does not increase mutual information, we obtain
\begin{IEEEeqnarray}{lCl}
\IEEEeqnarraymulticol{3}{l}{I\bigl(X_k;Y_{r,1}^n\bigm| X_1^{k-1}\bigr)} \nonumber\\
\quad & \geq & I\bigl(X_k;Y_{r,1}^k\bigm| X_1^{k-1}\bigr) \nonumber\\
& = & I\bigl(X_k;Y_{r,1}^k,X_1^{k-1}\bigr) \nonumber\\
& \geq & I\bigl(X_k;Y_{r,k-\kappa}^k,X_{k-\kappa}^{k-1}\bigr) \nonumber\\
& = & I\bigl(X_k;Y_{r,k-\kappa}^k,H_{1,k-\kappa}^{k-1},X_{k-\kappa}^{k-1}\bigr) - \Delta_1(\SNR,\kappa) \label{eq:lb1_1}
\end{IEEEeqnarray}
where $\Delta_1$ is defined as
\begin{IEEEeqnarray}{lCl}
\Delta_1(\SNR,\kappa) & \triangleq & I\bigl(X_k;Y_{r,k-\kappa}^k,H_{1,k-\kappa}^{k-1},X_{k-\kappa}^{k-1}\bigr) \nonumber\\
& & {} - I\bigl(X_k;Y_{r,k-\kappa}^k,X_{k-\kappa}^{k-1}\bigr) \nonumber\\
& = & I\bigl(X_k;H_{1,k-\kappa}^{k-1}\bigm| Y_{r,k-\kappa}^k,X_{k-\kappa}^{k-1}\bigr).
\end{IEEEeqnarray}
Due to the stationarity of the channel and of the proposed coding scheme, $\Delta_1(\SNR,\kappa)$ does not depend on $k$. Furthermore, it follows from \cite[App.~IX]{lapidothmoser03_3} that for every fixed $\kappa$ we have
\begin{equation}
\label{eq:vareps1}
\lim_{\SNR\to\infty} \Delta_1(\SNR,\kappa) = 0.
\end{equation}
We further lower-bound the RHS of \eqref{eq:lb1_1} by
\begin{IEEEeqnarray}{lCl}
I\bigl(X_k;Y_{r,k-\kappa}^k,H_{1,k-\kappa}^{k-1},X_{k-\kappa}^{k-1}\bigr) & \geq & I\bigl(X_k;Y_{r,k},H_{1,k-\kappa}^{k-1}\bigr)\nonumber\\
& = & I\bigl(X_k;Y_{r,k}\bigm| H_{1,k-\kappa}^{k-1}\bigr) \IEEEeqnarraynumspace\label{eq:lb1_2}
\end{IEEEeqnarray}
which follows because reducing observations does not increase mutual information and because $X_k$ and $H_{1,k-\kappa}^{k-1}$ are independent.

We express the fading coefficient $H_{1,k}$ at time $k$ as
\begin{equation*}
H_{1,k} = \bar{H}_{1,k} + \tilde{H}_{1,k}
\end{equation*}
where $\bar{H}_{1,k}=\Expec\bigl[H_{1,k}\bigm| H_{1,k-\kappa}^{k-1}\bigr]$ is the minimum-mean-square-error predictor of $H_{1,k}$ given $H_{1,k-1},\ldots,H_{1,k-\kappa}$, and where $\tilde{H}_{1,k}$ denotes the prediction error. Note that since $\{H_{1,k},\,k\in\Integers\}$ is a zero-mean, complex Gaussian process, it follows that $\bar{H}_{1,k}$ and $\tilde{H}_{1,k}$ are zero-mean, complex Gaussian random variables with variance $1-\eps_{1,\kappa}^2$ and $\eps_{1,\kappa}^2$, respectively. Further note that $\tilde{H}_{1,k}$ is independent of $H_{1,k-\kappa}^{k-1}$ \cite[Lemma 5.8]{wienermasani57}, and that \cite{doob90}, \cite[Lemmas 5.7(b) \& 5.10(c)]{wienermasani57}
\begin{equation}
\lim_{\kappa\to\infty} \eps_{1,\kappa}^2 = \eps_1^2 = \exp\Biggl(\int_{-1/2}^{1/2} \log F_1'(\lambda)\d\lambda\Biggr). \label{eq:noisy_pred}
\end{equation}
With this, we obtain
\begin{IEEEeqnarray}{lCl}
\IEEEeqnarraymulticol{3}{l}{I\bigl(X_k;Y_{r,k}\bigm| H_{1,k-\kappa}^{k-1}\bigr)}\nonumber\\
\quad & = & h\Bigl(\bigl(\bar{H}_{1,k}+\tilde{H}_{1,k}\bigr)X_k+ Z_{r,k}\Bigm| \bar{H}_{1,k}\Bigr) \nonumber\\
& & {} - h\Bigl(\bigl(\bar{H}_{1,k}+\tilde{H}_{1,k}\bigr)X_k+ Z_{r,k}\Bigm| \bar{H}_{1,k},X_k\Bigr) \nonumber\\
& \geq & h\bigl(H_{1,k} X_k+ Z_{r,k}\bigm| H_{1,k},Z_{r,k}\bigr)\nonumber\\
& & {} - h\Bigl(\bigl(\bar{H}_{1,k}+\tilde{H}_{1,k}\bigr)X_k+ Z_{r,k}\Bigm| \bar{H}_{1,k},X_k\Bigr) \label{eq:75}
\end{IEEEeqnarray}
where we have used that conditioning does not increase entropy. From the behavior of differential entropy under translation and scaling by a complex number \cite[Ths.~9.6.3 \& 9.6.4]{coverthomas91}, it follows that
\begin{IEEEeqnarray}{lCl}
\IEEEeqnarraymulticol{3}{l}{h\bigl(H_{1,k} X_k+ Z_{r,k}\bigm| H_{1,k},Z_{r,k}\bigr)} \nonumber\\
\quad & = & h\bigl(H_{1,k}X_k\bigm| H_{1,k}\bigr) \nonumber\\
& = & \E{\log|H_{1,k}|^2} + h(X_k)
\end{IEEEeqnarray}
and
\begin{IEEEeqnarray}{lCl}
\IEEEeqnarraymulticol{3}{l}{h\Bigl(\bigl(\bar{H}_{1,k}+\tilde{H}_{1,k}\bigr)X_k+ Z_{r,k}\Bigm| \bar{H}_{1,k},X_k\Bigr)} \nonumber\\
\quad & = & h\bigl(\tilde{H}_{1,k}X_k + Z_{r,k}\bigm| X_k\bigr) \nonumber\\
& = & \E{\log|X_k|^2} + h\biggl(\tilde{H}_{1,k}+\frac{Z_{r,k}}{X_k}\biggm| X_k\biggr).
\end{IEEEeqnarray}
Since, conditioned on $X_k$, the random variable $\tilde{H}_{1,k}+Z_{r,k}/X_k$ is Gaussian, we have
\begin{IEEEeqnarray}{lCl}
\IEEEeqnarraymulticol{3}{l}{h\biggl(\tilde{H}_{1,k}+\frac{Z_{r,k}}{X_k}\biggm| X_k\biggr)}\nonumber\\
\quad & = &  \log(\pi e) + \E{\log\biggl(\eps^2_{1,\kappa}+\frac{\sigma^2}{|X_k|^2}\biggr)} \nonumber\\
& \leq &  \log(\pi e) + \log\biggl(\eps^2_{1,\kappa}+\frac{\sigma^2}{\log\const{A}^2}\biggr) \label{eq:lb1_3_new}
\end{IEEEeqnarray}
where the last step follows because for our choice of input distribution \eqref{eq:DFindep1} we have $|X_k|^2\geq\log\const{A}^2$ with probability one.

Combining \eqref{eq:75}--\eqref{eq:lb1_3_new} yields
\begin{IEEEeqnarray}{lCl}
\IEEEeqnarraymulticol{3}{l}{I\bigl(X_k;Y_{r,k}\bigm| H_{1,k-\kappa}^{k-1}\bigr)}\nonumber\\
\quad & \geq & \E{\log|H_{1,k}|^2} + h(X_k) - \E{\log|X_k|^2} - \log\pi \nonumber\\
& & {}  - 1 - \log\biggl(\eps^2_{1,\kappa}+\frac{\sigma^2}{\log\const{A}^2}\biggr). \label{eq:lb1_3}
\end{IEEEeqnarray}
The first term on the RHS of \eqref{eq:lb1_3} evaluates to \cite[Sec.~4.331]{gradshteynryzhik00}
\begin{equation}
\label{eq:lb1_4}
\E{\log|H_{1,k}|^2} = -\gamma.
\end{equation}
The subsequent three terms yield \cite[Lemmas 6.15 \& 6.16]{lapidothmoser03_3}
\begin{IEEEeqnarray}{lCl}
\IEEEeqnarraymulticol{3}{l}{h(X_k) - \E{\log|X_k|^2} - \log\pi} \nonumber\\
\quad & = & h\bigl(\log|X_k|^2\bigr) \nonumber\\
& = & \log\bigl(\alpha \log\const{A}^2-\log\log\const{A}^2\bigr).\label{eq:lb1_5}
\end{IEEEeqnarray}
Combining \eqref{eq:lb1_1}--\eqref{eq:lb1_5}, and noting that the RHS of \eqref{eq:lb1_3} does not depend on $k$, we obtain
\begin{IEEEeqnarray}{lCll}
\IEEEeqnarraymulticol{4}{l}{\frac{1}{n} I\bigl(X_1^n;Y_{r,1}^n\bigm|X_{r,1}^n\bigr)}\nonumber\\
\quad & \geq & \frac{n-\kappa}{n} \Biggl[&\log\bigl(\alpha\log\const{A}^2-\log\log\const{A}^2\bigr) -1-\gamma \nonumber\\
& & & \, {} - \log\biggl(\eps^2_{1,\kappa}+\frac{\sigma^2}{\log\const{A}^2}\biggr)-\Delta_1(\SNR,\kappa)\Biggr]\IEEEeqnarraynumspace
\end{IEEEeqnarray}
which tends to
\begin{IEEEeqnarray}{l}
\log\bigl(\alpha\log\const{A}^2-\log\log\const{A}^2\bigr)-1-\gamma \nonumber\\
\quad {} -\log\biggl(\eps^2_{1,\kappa}+\frac{\sigma^2}{\log\const{A}^2}\biggr)-\Delta_1(\SNR,\kappa)
\end{IEEEeqnarray}
as $n$ tends to infinity. Using \eqref{eq:vareps1}, and noting that for every fixed $\rho>0$ we have
\begin{IEEEeqnarray}{rCl}
\IEEEeqnarraymulticol{3}{l}{\lim_{\SNR\to\infty} \bigl\{\log\bigl(\alpha\log\const{A}^2-\log\log\const{A}^2\bigr)-\log\log\SNR\bigr\}}\nonumber\\
\quad\qquad\qquad\qquad\qquad\qquad\qquad\qquad\qquad\qquad & = & \log\alpha \IEEEeqnarraynumspace
\end{IEEEeqnarray}
and
\begin{IEEEeqnarray}{rCl}
\lim_{\SNR\to\infty} \log\biggl(\eps^2_{1,\kappa}+\frac{\sigma^2}{\log\const{A}^2}\biggr) & = & \log\eps^2_{1,\kappa}
\end{IEEEeqnarray}
we obtain
\begin{IEEEeqnarray}{lCl}
\IEEEeqnarraymulticol{3}{l}{\varlimsup_{\SNR\to\infty}\biggl\{\lim_{n\to\infty} \frac{1}{n}  I\bigl(X_1^n;Y_{r,1}^n\bigm|X_{r,1}^n\bigr)-\log\log\SNR\biggr\}} \nonumber\\
\quad & \geq & {} -1-\gamma+\log\frac{1}{\eps_{1,\kappa}^2}+\log\alpha.
\end{IEEEeqnarray}
By \eqref{eq:noisy_pred}, this tends to
\begin{IEEEeqnarray}{lCl}
\IEEEeqnarraymulticol{3}{l}{\varlimsup_{\SNR\to\infty}\biggl\{\lim_{n\to\infty} \frac{1}{n}  I\bigl(X_1^n;Y_{r,1}^n\bigm|X_{r,1}^n\bigr)-\log\log\SNR\biggr\}} \nonumber\\
\quad & \geq & {} -1-\gamma+\log\frac{1}{\eps_{1}^2}+\log\alpha \label{eq:lb1_LB1}
\end{IEEEeqnarray}
as $\kappa$ tends to infinity.

\subsection{Lower bound on $\lim_{n\to\infty}\frac{1}{n} I\bigl(X_1^n,X_{r,1}^n;Y_1^n\bigr)$}
We continue by lower-bounding the second term on the RHS of \eqref{eq:propDF}. The proof is similar to the proof of \eqref{eq:lb1_LB1}, and we will therefore skip some of the details. We start with the chain rule for mutual information to obtain
\begin{IEEEeqnarray}{lCl}
\IEEEeqnarraymulticol{3}{l}{I\bigl(X_1^n,X_{r,1}^n;Y_1^n\bigr)} \nonumber\\
\quad & = & \sum_{k=1}^n I\bigl(X_k,X_{r,k};Y_1^n\bigm|X_1^{k-1},X_{r,1}^{k-1}\bigr) \nonumber\\
& = & \sum_{k=1}^n I\bigl(X_k,X_{r,k};Y_1^n,X_1^{k-1},X_{r,1}^{k-1}\bigr) \nonumber\\
& \geq & \sum_{k=\kappa+1}^n I\bigl(X_k,X_{r,k};Y_{k-\kappa}^k,X_{k-\kappa}^{k-1},X_{r,k-\kappa}^{k-1}\bigr)\label{eq:lb1_a}
\end{IEEEeqnarray}
for some arbitrary $0\leq\kappa<n$. Here the second step follows because $\{(X_k,X_{r,k}),\,k\in\Integers\}$ is i.i.d.  We next define $\Delta_2(\SNR,\kappa)$ as
\begin{IEEEeqnarray}{lCl}
\IEEEeqnarraymulticol{3}{l}{\Delta_2(\SNR,\kappa)} \nonumber\\
\quad & \triangleq & I\bigl(X_k,X_{r,k};Y_{k-\kappa}^k,X_{k-\kappa}^{k-1},X_{r,k-\kappa}^{k-1},H_{3,k-\kappa}^{k-1}\bigr) \nonumber\\
& & {}  - I\bigl(X_k,X_{r,k};Y_{k-\kappa}^k,X_{k-\kappa}^{k-1},X_{r,k-\kappa}^{k-1}\bigr) \nonumber\\
& = & I\bigl(X_k,X_{r,k};H_{3,k-\kappa}^{k-1} \bigm| Y_{k-\kappa}^k,X_{k-\kappa}^{k-1},X_{r,k-\kappa}^{k-1}\bigr) 
\end{IEEEeqnarray}
for which we show in Appendix~\ref{app:vareps2} that, for every fixed $\kappa$,
\begin{equation}
\label{eq:vareps2}
\lim_{\SNR\to\infty} \Delta_2(\SNR,\kappa) = 0.
\end{equation}
With this definition, every summand on the RHS of \eqref{eq:lb1_a} can be lower-bounded by
\begin{IEEEeqnarray}{lCl}
\IEEEeqnarraymulticol{3}{l}{I\bigl(X_k,X_{r,k};Y_{k-\kappa}^k,X_{k-\kappa}^{k-1},X_{r,k-\kappa}^{k-1}\bigr)} \nonumber\\
& = & I\bigl(X_k,X_{r,k};Y_{k-\kappa}^k,X_{k-\kappa}^{k-1},X_{r,k-\kappa}^{k-1},H_{3,k-\kappa}^{k-1}\bigr) -\Delta_2(\SNR,\kappa)\nonumber\\
& = & I\bigl(X_k,X_{r,k};Y_{k-\kappa}^k,X_{k-\kappa}^{k-1},X_{r,k-\kappa}^{k-1}\bigm| H_{3,k-\kappa}^{k-1}\bigr) \nonumber\\
& & {} - \Delta_2(\SNR,\kappa)\nonumber\\
& \geq & I\bigl(X_k,X_{r,k};Y_k \bigm| H_{3,k-\kappa}^{k-1}\bigr) - \Delta_2(\SNR,\kappa) \label{eq:lb1_b}
\end{IEEEeqnarray}
where the second step follows because $\{H_{3,k},\,k\in\Integers\}$ is independent of $(X_k,X_{r,k})$; the third step follows because reducing observations does not increase mutual information.

As above, we express the fading $H_{3,k}$ as
\begin{equation*}
H_{3,k} = \bar{H}_{3,k} + \tilde{H}_{r,k}
\end{equation*}
where $\bar{H}_{3,k}=\Expec\bigl[H_{3,k}\bigm| H_{3,k-\kappa}^{k-1}\bigr]$ and where $\tilde{H}_{3,k}$ is a zero-mean, complex Gaussian random variable of variance $\eps_{3,\kappa}^2$ satisfying
\begin{equation}
\lim_{\kappa\to\infty} \eps_{3,\kappa}^2 = \eps_3^2 = \exp\Biggl(\int_{-1/2}^{1/2}\log F_3'(\lambda)\d\lambda\Biggr). \label{eq:dobi_eps3}
\end{equation}
We thus have
\begin{IEEEeqnarray}{lCl}
\IEEEeqnarraymulticol{3}{l}{I\bigl(X_k,X_{r,k};Y_{k} \bigm| H_{3,k-\kappa}^{k-1}\bigr)} \nonumber\\
\quad & = & h\bigl(Y_k\bigm|H_{3,k-\kappa}^{k-1}\bigr) - h\bigl(Y_k\bigm|X_k,X_{r,k},H_{3,k-\kappa}^{k-1}\bigr) \nonumber\\
& \geq & h\bigl(Y_k\bigm|H_{3,k-\kappa}^{k-1},H_{3,k},H_{2,k},X_k,Z_k\bigr) \nonumber\\
& & {} - h\bigl(Y_k\bigm|X_k,X_{r,k},\bar{H}_{3,k}\bigr) \nonumber\\
& = & h\bigl(H_{3,k} X_{r,k}\bigm|H_{3,k}\bigr) \nonumber\\
& & {} - h\bigl(\tilde{H}_{3,k} X_{r,k}+H_{2,k}X_k+Z_k\bigm| X_k,X_{r,k}\bigr) \IEEEeqnarraynumspace \label{eq:lb1_c}
\end{IEEEeqnarray}
where the second step follows because conditioning does not increase entropy and because $\bar{H}_{3,k}$ is a function of $H_{3,k-\kappa}^{k-1}$; and the last step follows from the property of differential entropy under translation and because the random variable $H_{3,k}X_{r,k}$ is independent of $\bigl(H_{3,k-\kappa}^{k-1},H_{2,k},X_k,Z_k\bigr)$ conditioned on $H_{3,k}$.

As above, we use the behavior of differential entropy under scaling by a complex number to evaluate the entropies on the RHS of \eqref{eq:lb1_c} as
\begin{IEEEeqnarray}{lCl}
h\bigl(H_{3,k} X_{r,k}\bigm|H_{3,k}\bigr) & = & \E{\log|H_{3,k}|^2} + h(X_{r,k})
\end{IEEEeqnarray}
and
\begin{IEEEeqnarray}{lCl}
\IEEEeqnarraymulticol{3}{l}{h\bigl(\tilde{H}_{3,k} X_{r,k}+H_{2,k}X_k+Z_k\bigm| X_k,X_{r,k}\bigr) = \E{\log|X_{r,k}|^2}} \nonumber\\
\IEEEeqnarraymulticol{3}{r}{ {} + h\biggl(\tilde{H}_{3,k}+H_{1,k}\frac{X_k}{X_{r,k}}+\frac{Z_k}{X_{r,k}}\biggm| X_k,X_{r,k}\biggr). \IEEEeqnarraynumspace}
\end{IEEEeqnarray}
Conditioned on $(X_k,X_{r,k})$, the random variable \[\tilde{H}_{3,k}+H_{1,k} X_k/X_{r,k}+Z_k/X_{r,k}\] is complex Gaussian, and we obtain
\begin{IEEEeqnarray}{lCl}
\IEEEeqnarraymulticol{3}{l}{h\biggl(\tilde{H}_{3,k}+H_{1,k}\frac{X_k}{X_{r,k}}+\frac{Z_k}{X_{r,k}}\biggm| X_k,X_{r,k}\biggr)} \nonumber\\
\,\, & = & \log(\pi e) + \E{\log\biggl(\eps^2_{3,\kappa}+\frac{|X_k|^2}{|X_{r,k}|^2}+\frac{\sigma^2}{|X_{r,k}|^2}\biggr)} \nonumber\\
& \leq & \log(\pi e) + \log\biggl(\eps^2_{3,\kappa}+\frac{\const{A}^{2\alpha}}{\const{A}_r^{2\beta}}+\frac{\sigma^2}{\const{A}_r^{2\beta}}\biggr) \nonumber\\
& = & \log(\pi e) +  \log\biggl(\eps^2_{3,\kappa}+\frac{1}{\rho^{2\beta}\const{A}^{2(\beta-\alpha)}}+\frac{\sigma^2}{\rho^{2\beta}\const{A}^{2\beta}}\biggr) \IEEEeqnarraynumspace\label{eq:93}
\end{IEEEeqnarray}
where the inequality follows because for our choice of input distribution \eqref{eq:DFindep1} and \eqref{eq:DFindep2} we have $|X_k|^2\leq \const{A}^{2\alpha}$ and \mbox{$|X_{r,k}|^2 \geq \const{A}_r^{2\beta}$} with probability one; the last equality follows from \eqref{eq:alpha}.

Combining \eqref{eq:lb1_c}--\eqref{eq:93} yields
\begin{IEEEeqnarray}{lCl}
\IEEEeqnarraymulticol{3}{l}{I\bigl(X_k,X_{r,k};Y_{k} \bigm| H_{3,k-\kappa}^{k-1}\bigr)} \nonumber\\
\quad & \geq & \E{\log|H_{3,k}|^2} + h(X_{r,k}) - \E{\log|X_{r,k}|^2} - \log\pi \nonumber\\
& & {} -1 - \log\biggl(\eps^2_{3,\kappa}+\frac{1}{\rho^{2\beta}\const{A}^{2(\beta-\alpha)}}+\frac{\sigma^2}{\rho^{2\beta}\const{A}^{2\beta}}\biggr) \nonumber\\
& = & \log\bigl(\log\const{A}_r^2-\beta \log \const{A}_r^2\bigr) - \gamma \nonumber\\
& & {} -1 - \log\biggl(\eps^2_{3,\kappa}+\frac{1}{\rho^{2\beta}\const{A}^{2(\beta-\alpha)}}+\frac{\sigma^2}{\rho^{2\beta}\const{A}^{2\beta}}\biggr) \IEEEeqnarraynumspace\label{eq:lb1_d}
\end{IEEEeqnarray}
where the last step follows by evaluating the first four terms as in \eqref{eq:lb1_4} and \eqref{eq:lb1_5}.

Combining \eqref{eq:lb1_a}--\eqref{eq:lb1_d}, and noting that the RHS of \eqref{eq:lb1_d} does not depend on $k$, we obtain
\begin{IEEEeqnarray}{lCll}
\IEEEeqnarraymulticol{4}{l}{\frac{1}{n} \bigl(X_1^n,X_{r,1}^n;Y_1^n\bigr)} \nonumber\\
 \quad & \geq & \frac{n-\kappa}{n} \Biggl[& \log\log\bigl(\rho^2\const{A}^2\bigr) + \log(1-\beta) \nonumber\\
 & & & \, {}  -\gamma-1 - \Delta_2(\SNR,\kappa)\nonumber\\
& & & \, {} - \log\biggl(\eps^2_{3,\kappa}+\frac{1}{\rho^{2\beta}\const{A}^{2(\beta-\alpha)}}+\frac{\sigma^2}{\rho^{2\beta}\const{A}^{2\beta}}\biggr)\Biggr]\IEEEeqnarraynumspace
\end{IEEEeqnarray}
which tends to
\begin{IEEEeqnarray}{l}
\log\log\bigl(\rho^2\const{A}^2\bigr) + \log(1-\beta) -\gamma-1 \nonumber\\ {}- \log\biggl(\eps^2_{3,\kappa}+\frac{1}{\rho^{2\beta}\const{A}^{2(\beta-\alpha)}}+\frac{\sigma^2}{\rho^{2\beta}\const{A}^{2\beta}}\biggr) - \Delta_2(\SNR,\kappa) \IEEEeqnarraynumspace
\end{IEEEeqnarray}
as $n$ tends to infinity. Using \eqref{eq:vareps2}, and noting that for every $0<\alpha<\beta<1$ and for every fixed $\rho>0$ we have
\begin{IEEEeqnarray}{rCl}
\lim_{\SNR\to\infty} \bigl\{\log\log\bigl(\rho^2\const{A}^2\bigr) -\log\log\SNR\bigr\} & = & 0
\end{IEEEeqnarray}
and
\begin{IEEEeqnarray}{lCl}
\IEEEeqnarraymulticol{3}{l}{\lim_{\SNR\to\infty} \log\biggl(\eps^2_{3,\kappa}+\frac{1}{\rho^{2\beta}\const{A}^{2(\beta-\alpha)}}+\frac{\sigma^2}{\rho^{2\beta}\const{A}^{2\beta}}\biggr)} \nonumber\\
\qquad\qquad\qquad\qquad\qquad\qquad\qquad\qquad & = & \log\eps_{3,\kappa}^2 \IEEEeqnarraynumspace
\end{IEEEeqnarray}
we obtain
\begin{IEEEeqnarray}{lCl}
\IEEEeqnarraymulticol{3}{l}{\varlimsup_{\SNR\to\infty} \biggl\{\lim_{n\to\infty}\frac{1}{n} \bigl(X_1^n,X_{r,1}^n;Y_1^n\bigr) - \log\log\SNR\biggr\}} \nonumber\\
\quad\qquad\qquad\qquad & \geq & {} - 1 - \gamma + \log\frac{1}{\eps_{3,\kappa}^2} + \log(1-\beta). \IEEEeqnarraynumspace
\end{IEEEeqnarray}
By \eqref{eq:dobi_eps3}, this tends to
\begin{IEEEeqnarray}{lCl}
\IEEEeqnarraymulticol{3}{l}{\varlimsup_{\SNR\to\infty} \biggl\{\lim_{n\to\infty}\frac{1}{n} \bigl(X_1^n,X_{r,1}^n;Y_1^n\bigr) - \log\log\SNR\biggr\}} \nonumber\\
\quad & \geq & - 1 - \gamma + \log\frac{1}{\eps_{3}^2} + \log(1-\beta) \label{eq:lb1_LB2}
\end{IEEEeqnarray}
as we let $\kappa$ tend to infinity.

\subsection{Maximizing over $\alpha$ and $\beta$}
It follows from \eqref{eq:propDF}, \eqref{eq:lb1_LB1}, and \eqref{eq:lb1_LB2} that a decode-and-forward strategy can achieve the fading number
\begin{IEEEeqnarray}{lCll}
\chi & \geq & \min\biggl\{& -1-\gamma+\log\frac{1}{\eps_1^2}+\log\alpha, \nonumber\\
& & & \, {} -1-\gamma+\log\frac{1}{\eps_3^2}+\log(1-\beta)\biggr\} \label{eq:lb1_beforemax}
\end{IEEEeqnarray}
for every $0<\alpha<\beta<1$. We next prove Theorem~\ref{thm:lowerbound1} by maximizing over $\alpha$ and $\beta$. Note that the first argument of the minimum in \eqref{eq:lb1_beforemax} is increasing in $\alpha$, whereas the second argument is decreasing in $\beta$. Consequently, we have
\begin{IEEEeqnarray}{lCl}
\IEEEeqnarraymulticol{3}{l}{\sup_{0<\alpha<\beta<1}\min\biggl\{-1-\gamma+\log\frac{1}{\eps_1^2}+\log\alpha,}\nonumber\\
\IEEEeqnarraymulticol{3}{l}{\,\,\qquad\qquad\qquad {} -1-\gamma+\log\frac{1}{\eps_3^2}+\log(1-\beta)\biggr\}}\nonumber\\
\quad & = & \max_{0<\alpha<1} \min\biggl\{-1-\gamma+\log\frac{1}{\eps_1^2}+\log\alpha, \nonumber\\
 & & \,\,\quad\qquad\qquad {} -1-\gamma+\log\frac{1}{\eps_3^2}+\log(1-\alpha)\biggr\} \IEEEeqnarraynumspace\label{eq:lb1_getridofbeta}
\end{IEEEeqnarray}
where the maximum on the RHS of \eqref{eq:lb1_getridofbeta} exists because \mbox{$x\mapsto \log(x)$} is continuous on $0<x<1$ and because $\alpha=0$ or $\alpha=1$ would imply that the RHS of \eqref{eq:lb1_getridofbeta} is $-\infty$, which is clearly suboptimal.

We next note that the first argument of the minimum in \eqref{eq:lb1_getridofbeta} is increasing in $\alpha$, whereas the second argument is decreasing in $\alpha$. Consequently, the optimal $\alpha$ must satisfy
\begin{equation}
\label{eq:alpha=beta}
-1-\gamma+\log\frac{1}{\eps^2_1}+\log\alpha = -1-\gamma+\log\frac{1}{\eps_3^2} + \log(1-\alpha).
\end{equation}
Solving \eqref{eq:alpha=beta} yields
\begin{equation}
\label{eq:optalpha}
\alpha = \frac{\eps_1^2}{\eps_1^2+\eps_3^2}
\end{equation}
which combined with \eqref{eq:lb1_beforemax} proves Theorem~\ref{thm:lowerbound1}.

\section{Quantize Map and Forward}
\label{sec:discussion}
Recently, a strategy called \emph{quantize-map-and-forward} was introduced by Avestimehr et al. \cite{avestimehrdiggavitse11}. They showed that this scheme achieves rates that are within a constant gap of the max-flow min-cut upper bound, where the gap depends on the number of relays but not on the channel parameters. For example, for the Gaussian relay channel with a single relay, and for the two-relay Gaussian diamond network, the gap is not more than one bit.

However, for the Gaussian relay channel with a single relay, rates that are within one bit of the max-flow min-cut upper bound can also be achieved by decode-and-forward \cite[Th.~3.1]{avestimehrdiggavitse11}. We therefore believe that for the above fading relay channel quantize-map-and-forward will give rates that are comparable to the ones presented in Theorem~\ref{thm:lowerbound1}. (For fading relay channels with more than one relay, quantize-map-and-forward may be superior to decode-and-forward.) Indeed, if the fading coefficient of the channel between the transmitter and the relay can be predicted more accurately from its infinite past than the fading coefficient of the channel between the relay and the receiver, then at high SNR decode-and-forward achieves rates that are within one bit of capacity (Corollary~\ref{cor:1bit}). If the fading coefficient of the channel between the transmitter and the relay cannot be predicted more accurately than the fading coefficient of the channel between the relay and the receiver, then the gap between the upper bound \eqref{eq:thmUB} and the lower bound \eqref{eq:thmlowerbound1} may be larger than one bit.

\section{Conclusion}
\label{sec:conclusion}
We have studied the capacity of noncoherent fading relay channels, where all terminals are aware of the statistics of the fading but not of their realizations. We demonstrated that, if the fading coefficient of the channel between the transmitter and the receiver can be predicted more accurately from its infinite past than the fading coefficient of the channel between the relay and the receiver, then direct communication achieves the fading number. We further showed that if the fading coefficient of the channel between the transmitter and the relay can be predicted more accurately from its infinite past than the fading coefficient of the channel between the relay and the receiver, then the fading number of the relay channel is within one bit of the fading number of the TRC-MISO fading channel.


\appendices

\section{Proof of Proposition~\ref{prop:nonasymptotic}}
\label{app:nonasymptotic}
Proposition~\ref{prop:nonasymptotic} follows by combining the \emph{asymptotic} lower bounds on the capacity of noncoherent fading \emph{relay} channels (see Theorem~\ref{thm:lowerbound1}) with the \emph{nonasymptotic} lower bounds on the capacity of \emph{point-to-point} noncoherent fading channels \cite[Prop.~4.1]{koch09}. The proof of Proposition~\ref{prop:nonasymptotic} is thus very similar to the proof of Theorem~\ref{thm:lowerbound1}. For completeness, we repeat the main arguments here.

To prove Proposition~\ref{prop:nonasymptotic}, we use a decode-and-forward strategy (Proposition~\ref{prop:decodeforward}) and evaluate \eqref{eq:propDF}, namely,
\begin{IEEEeqnarray}{lCll}
R & = &  \lim_{n\to\infty} \sup \frac{1}{n} \min\Bigl\{& I\bigl(X_1^n;Y_{r,1}^n\bigm| X_{r,1}^n\bigr), \nonumber\\
& & & \quad I\bigl(X_1^n,X_{r,1}^n;Y_1^n\bigr)\Bigr\} \label{eq:appDF}
\end{IEEEeqnarray}
 for $\{X_k,\,k\in\Integers\}$ and $\{X_{r,k},\,k\in\Integers\}$ being i.i.d., circularly-symmetric, complex random variables, independent of each other and with
\begin{IEEEeqnarray}{lCll}
\log|X_k|^2 & \sim & \Uniform{\bigl[\alpha\log\bigl(\delta^2\const{A}^2\bigr),\alpha\log\const{A}^2\bigr]}, \quad & k\in\Integers\IEEEeqnarraynumspace\label{eq:85}\\
\log|X_{r,k}|^2 & \sim & \Uniform{\bigl[\log\bigl(\delta_r^2\const{A}_r^2\bigr),\log\const{A}_r^2\bigr]}, \quad & k\in\Integers \label{eq:86}
\end{IEEEeqnarray}
where $0<\alpha,\delta,\delta_r<1$. Note that \eqref{eq:85} and \eqref{eq:86} are generalizations of the input distributions \eqref{eq:DFindep1} and \eqref{eq:DFindep2} used to prove Theorem~\ref{thm:lowerbound1}. Both pairs of distributions are equal when
\begin{IEEEeqnarray}{lCl}
\delta & = & \frac{\bigl(\log\const{A}^2\bigr)^{\frac{1}{2\alpha}}}{\const{A}} \\
\delta_r & = & \const{A}_r^{-2(1-\beta)}.
\end{IEEEeqnarray}

\subsection{Lower bound on $\lim_{n\to\infty} \frac{1}{n} I\bigl(X_1^n;Y_{r,1}^n\bigm| X_{r,1}^n\bigr)$}
We evaluate the first term on the RHS of \eqref{eq:appDF} by using independent $\{X_k,\,k\in\Integers\}$ and $\{X_{r,k},\,k\in\Integers\}$ so that
\begin{equation}
\label{eq:appnonasymptotic1}
\lim_{n\to\infty} \frac{1}{n} I\bigl(X_1^n;Y_{r,1}^n\bigm| X_{r,1}^n\bigr) = \lim_{n\to\infty} \frac{1}{n} I\bigl(X_1^n;Y_{r,1}^n\bigr).
\end{equation}
The RHS of \eqref{eq:appnonasymptotic1} corresponds to the rates achievable over the point-to-point fading channel between the transmitter and the relay. Since \eqref{eq:85} is the distribution used to prove Proposition~4.1 in \cite{koch09} (provided that we replace $\const{A}^2$ and $\alpha$ in \cite{koch09} by $\const{A}^{2\alpha}$ and $\delta^\alpha$) it follows from \cite[Prop.~4.1]{koch09} and the definition of the SNR that
\begin{IEEEeqnarray}{lCl}
\IEEEeqnarraymulticol{3}{l}{\lim_{n\to\infty} \frac{1}{n} I\bigl(X_1^n;Y_{r,1}^n\bigr)} \nonumber\\
\,\,\, & \geq & \log\biggl(\frac{\sigma^{2(1-\alpha)}}{\delta^{\alpha}\SNR^{\alpha}}\biggr) -\int_{-1/2}^{1/2}\log\biggl(F_1'(\lambda)+\frac{\sigma^{2(1-\alpha)}}{\delta^{2\alpha}\SNR^{\alpha}}\biggr)\d\lambda \nonumber\\
& & {} - \exp\Biggl(\frac{\sigma^{2(1-\alpha)} e}{\alpha\log\bigl(\frac{1}{\delta^2}\bigr)\delta^{\alpha}\SNR^{\alpha}}\Biggr)\times \nonumber\\
& & \quad\qquad\times \Ei{-\frac{\sigma^{2(1-\alpha)} e}{\alpha\log\bigl(\frac{1}{\delta^2}\bigr)\delta^{\alpha}\SNR^{\alpha}}}. \label{eq:appnonasymptotic_L1}\IEEEeqnarraynumspace
\end{IEEEeqnarray}

\subsection{Lower bound on $\lim_{n\to\infty} \frac{1}{n} I\bigl(X_1^n,X_{r,1}^n;Y_1^n\bigr)$}
We next lower-bound the second term on the RHS of \eqref{eq:appDF}. Since reducing observations cannot increase mutual information, we have
\begin{IEEEeqnarray}{lCl}
\IEEEeqnarraymulticol{3}{l}{\lim_{n\to\infty} \frac{1}{n} I\bigl(X_1^n,X_{r,1}^n;Y_1^n\bigr)}\nonumber\\
\quad & = & \lim_{n\to\infty} \frac{1}{n} \sum_{k=1}^n I\bigl(X_k,X_{r,k};Y_1^n\bigm| X_1^{k-1},X_{r,1}^{k-1}\bigr) \nonumber\\
& \geq & \lim_{n\to\infty}\frac{1}{n} \sum_{k=1}^n I\bigl(X_k,X_{r,k};Y_1^k\bigm| X_1^{k-1},X_{r,1}^{k-1}\bigr). \IEEEeqnarraynumspace
\end{IEEEeqnarray}
Using a Ces\'aro-type theorem \cite[Th.~4.2.3]{coverthomas91}, we obtain
\begin{IEEEeqnarray}{lCl}
\IEEEeqnarraymulticol{3}{l}{\lim_{n\to\infty}\frac{1}{n} \sum_{k=1}^n I\bigl(X_k,X_{r,k};Y_1^k\bigm| X_1^{k-1},X_{r,1}^{k-1}\bigr)}\nonumber\\
\,\, & \geq & \varliminf_{k\to\infty} I\bigl(X_k,X_{r,k};Y_1^k\bigm| X_1^{k-1},X_{r,1}^{k-1}\bigr) \nonumber\\
& \geq & \varliminf_{k\to\infty} I\Biggl(X_k,X_{r,k};Y_k,\biggl\{\frac{Y_{\ell}}{X_{r,\ell}}\biggr\}_{\ell=1}^{k-1}\Biggm| X_1^{k-1},X_{r,1}^{k-1}\Biggr) \IEEEeqnarraynumspace \label{eq:appnonasymptotic2}
\end{IEEEeqnarray}
where $\varliminf$ denotes the \emph{limit inferior}. We next minimize the RHS of \eqref{eq:appnonasymptotic2} over all $\bigl(x_1^{k-1},x_{r,1}^{k-1}\bigr)$ satisfying
\begin{IEEEeqnarray}{lCcCll}
\delta^{2\alpha} \const{A}^{2\alpha} & \leq & |x_{\ell}|^2 & \leq & \const{A}^{2\alpha}, \quad &\ell=1,\ldots,k-1 \label{eq:stuff_satisfying1} \\
\delta_r^2\const{A}_r^2 & \leq & |x_{r,\ell}|^2 & \leq & \const{A}^2_r, \quad &\ell=1,\ldots,k-1.\label{eq:stuff_satisfying2}
\end{IEEEeqnarray}
Specifically, we show that the RHS of \eqref{eq:appnonasymptotic2} is minimized when
\begin{equation}
|x_{\ell}|^2=\const{A}^{2\alpha} \quad \textnormal{and} \quad |x_{r,\ell}|^2=\delta_r^2\const{A}_r^2
\end{equation} 
for $\ell=1,\ldots,k-1$. Indeed, since $\{H_{2,k},\,k\in\Integers\}$ and $\{Z_k,\,k\in\Integers\}$ are both sequences of i.i.d.\ Gaussian random variables, it follows that, conditioned on \[\bigl(X_1^{k-1},X_{r,1}^{k-1}\bigr)=\bigl(x_1^{k-1},x_{r,1}^{k-1}\bigr)\] the random variables $Y_{\ell}/x_{r,\ell}$ satisfy
\begin{IEEEeqnarray}{lCl}
\frac{Y_{\ell}}{x_{r,\ell}} & = & H_{3,\ell}+H_{2,\ell} \frac{x_{\ell}}{x_{r,\ell}} + \frac{Z_{\ell}}{x_{r,\ell}} \nonumber\\
& \eqlaw & H_{3,\ell}+\biggl(\frac{x_{\ell}}{x_{r,\ell}}+\frac{\sigma^2}{x_{r,\ell}}\biggr) W_{\ell} \label{eq:appnonasymptotic3}
\end{IEEEeqnarray}
where $\{W_k,\,k\in\Integers\}$ is a sequence of i.i.d., zero-mean, unit-variance, complex Gaussian random variables, and where $A\eqlaw B$ indicates that $A$ and $B$ have the same law. The second term on the RHS of \eqref{eq:appnonasymptotic3} can be viewed as an additive-noise term. Thus, by choosing $|x_{\ell}|^2=\const{A}^{2\alpha}$ and $|x_{r,\ell}|^2=\delta_r^2\const{A}_r^2$, the variance of the additive noise is maximized. We next argue that maximizing the variance of the additive noise minimizes the mutual information. Indeed, suppose that the noise that minimizes the mutual information is not the one with maximum variance. Then we can add i.i.d.\ zero-mean Gaussian noise $\{U_k,\,k\in\Integers\}$ to $Y_{\ell}/x_{r,\ell}$ such that $Y_{\ell}/x_{r,\ell}+U_{\ell}$ has the same distribution as $Y_{\ell}/x_{r,\ell}$ when $|x_{\ell}|^2=\const{A}^{2\alpha}$ and $|x_{r,\ell}|^2=\delta_r^2\const{A}_r^2$. The claim follows by the Data Processing Inequality.

Using \eqref{eq:alpha}, we obtain from \eqref{eq:appnonasymptotic2} that
\begin{IEEEeqnarray}{lCl}
\IEEEeqnarraymulticol{3}{l}{\varliminf_{k\to\infty} I\Biggl(X_k,X_{r,k};Y_k,\biggl\{\frac{Y_{\ell}}{X_{r,\ell}}\biggr\}_{\ell=1}^{k-1}\Biggm| X_1^{k-1},X_{r,1}^{k-1}\Biggr)} \nonumber\\
\quad & \geq & \varliminf_{k\to\infty} I\Bigl(X_k,X_{r,k};Y_k, \bigl\{H_{3,\ell}+\xi W_{\ell}\bigr\}_{\ell=1}^{k-1}\Bigr) \nonumber\\
& = & \varliminf_{k\to\infty} I\Bigl(X_k,X_{r,k};Y_k \Bigm| \bigl\{H_{3,\ell}+\xi W_{\ell}\bigr\}_{\ell=1}^{k-1}\Bigr) \label{eq:appnonasymptotic4}
\end{IEEEeqnarray}
where
\begin{equation*}
\xi \triangleq \frac{1}{\delta_r^2\rho^2\const{A}^{2(1-\alpha)}}+\frac{\sigma^2}{\delta_r^2\rho^2\const{A}^2}.
\end{equation*}
Here, the first step follows by minimizing over all $\bigl(x_1^{k-1},x_{r,1}^{k-1}\bigr)$ satisfying \eqref{eq:stuff_satisfying1} and \eqref{eq:stuff_satisfying2} and because the joint law of $\bigl(X_k,X_{r,k},Y_k,\{H_{3,\ell}+\xi W_{\ell}\}_{\ell=1}^{k-1}\bigr)$
does not depend on $\bigl(x_1^{k-1},x_{r,1}^{k-1}\bigr)$; the last step follows because the pair $\bigl(X_k,X_{r,k}\bigr)$ is independent of $\bigl(\{H_{3,k},\,k\in\Integers\}, \{W_k,\,k\in\Integers\}\bigr)$.

We continue by expressing the mutual information as the difference of two differential entropies, i.e., we have
\begin{IEEEeqnarray}{lCl}
\IEEEeqnarraymulticol{3}{l}{I\Bigl(X_k,X_{r,k};Y_k \Bigm| \bigl\{H_{3,\ell}+\xi W_{\ell}\bigr\}_{\ell=1}^{k-1}\Bigr)}\nonumber\\
\quad & = &  h\Bigl(Y_k\Bigm| \bigl\{H_{3,\ell}+\xi W_{\ell}\bigr\}_{\ell=1}^{k-1}\Bigr)\nonumber\\
& & {} - h\Bigl(Y_k\Bigm|X_k,X_{r,k},\bigl\{H_{3,\ell}+\xi W_{\ell}\bigr\}_{\ell=1}^{k-1}\Bigr). \IEEEeqnarraynumspace \label{eq:appbla}
\end{IEEEeqnarray}
For the second entropy, it follows from the behavior of differential entropy under scaling by a complex number that
\begin{IEEEeqnarray}{lCl}
\IEEEeqnarraymulticol{3}{l}{h\Bigl(Y_k\Bigm|X_k,X_{r,k},\bigl\{H_{3,\ell}+\xi W_{\ell}\bigr\}_{\ell=1}^{k-1}\Bigr)} \nonumber\\
\quad & = & \E{\log|X_{r,k}|^2}\nonumber\\
& & {} +  h\biggl(\frac{Y_k}{X_{r,k}}\biggm| X_k,X_{r,k},\bigl\{H_{3,\ell}+\xi W_{\ell}\bigr\}_{\ell=1}^{k-1}\biggr) \nonumber\\
& \leq & \log\bigl(\delta_r\rho^2\const{A}^2\bigr)+ \log(\pi e) + \log\bigl(\eps_{3,k}^2(\xi)+\xi\bigr) \IEEEeqnarraynumspace \label{eq:appnonasymptotic5}
\end{IEEEeqnarray}
where $\eps_{3,k}^2(\xi)$ denotes the minimum-mean-square error in predicting $H_{3,k}$ from $(H_{3,k-1}+\xi W_{k-1}),\ldots,(H_{3,1}+\xi W_{1})$. Note that we have \cite[Sec.~III]{lapidoth05}
\begin{equation}
\label{eq:appnoisypred}
\lim_{k\to\infty} \eps_{3,k}^2(\xi) = \exp\Biggl(\int_{-1/2}^{1/2} \log\bigl(F_3'(\lambda)+\xi\bigr)\d\lambda\Biggr) - \xi.
\end{equation}
The inequality in \eqref{eq:appnonasymptotic5} follows by evaluating $\E{\log|X_{r,k}|^2}$ for the distribution \eqref{eq:86} and by noting that, conditioned on
\begin{equation*}
\Bigl(X_k,X_{r,k},\bigl\{H_{3,\ell}+\xi W_{\ell}\bigr\}_{\ell=1}^{k-1}\Bigr)
\end{equation*}
the random variable \[\frac{Y_k}{X_{r,k}} = H_{3,k}+H_{2,k}\frac{X_k}{X_{r,k}}+\frac{Z_k}{X_{r,k}}\] is complex Gaussian with variance
\begin{equation*}
\eps^2_{3,k}(\xi)+\frac{|X_k|^2}{|X_{r,k}|^2} + \frac{\sigma^2}{|X_{r,k}|^2}.
\end{equation*}
Since maximizing the differential entropy of a complex Gaussian random variable is tantamount to maximizing its variance, and since the variance is maximized for $|x_k|^2=\const{A}^{2\alpha}$ and $|x_{r,k}|^2=\delta_r^2\const{A}_r^2$, the inequality in \eqref{eq:appnonasymptotic5} follows.

For the first entropy on the RHS of \eqref{eq:appbla}, we have
\begin{IEEEeqnarray}{lCl}
\IEEEeqnarraymulticol{3}{l}{h\Bigl(Y_k\Bigm|\bigl\{H_{3,\ell}+\xi W_{\ell}\bigr\}_{\ell=1}^{k-1}\Bigr)} \nonumber\\
\,\, & \geq & h\bigl(H_{3,k} X_{r,k}+H_{2,k} X_k + Z_k \bigm| H_{3,k}\bigr) \nonumber\\
& \geq & \E{\log\biggl(e^{\log|H_{3,k}|^2+h(X_{r,k})}+e^{h(H_{2,k}X_k)}+\pi e\sigma^2\biggr)} \IEEEeqnarraynumspace\label{eq:appnonasymptotic6}
\end{IEEEeqnarray}
where the first step follows because conditioning does not increase entropy and because, conditioned on $H_{3,k}$, the channel output $Y_k$ is independent of $\{H_{3,\ell}+\xi W_{\ell}\}_{\ell=1}^{k-1}$; the second step follows from the Entropy Power Inequality \cite[Th.~16.6.3]{coverthomas91} and from the property of differential entropy under scaling by a complex number. 

Following the same steps as in \eqref{eq:lb1_5}, the differential entropy of $X_{r,k}$ for the distribution \eqref{eq:86} can be evaluated as
\begin{IEEEeqnarray}{lCl}
h(X_{r,k}) 
& = & \log\biggl(\log\biggl(\frac{1}{\delta_r^2}\biggr) \delta_r\const{A}_r^2\pi\biggr). \label{eq:appEP1}
\end{IEEEeqnarray}
We lower-bound $h(H_{2,k}X_k)$ by conditioning on $H_{2,k}$ and by using the property of differential entropy under scaling by a complex number:
\begin{equation}
h\bigl(H_{2,k}X_k\bigr) \geq \E{\log |H_{2,k}|^2} + h(X_k). \label{eq:120}
\end{equation}
By noting that $\E{\log |H_{2,k}|^2}=-\gamma$, cf.~\eqref{eq:lb1_4}, and following the same steps as in \eqref{eq:lb1_5} to evaluate the differential entropy of $X_k$ for the distribution \eqref{eq:85}, it follows that
\begin{equation}
h\bigl(H_{2,k}X_k\bigr) \geq \log\biggl(e^{-\gamma}\alpha\log\biggl(\frac{1}{\delta^2}\biggr)\delta^{\alpha}\const{A}^{2\alpha}\pi\biggr). \label{eq:appEP2}
\end{equation}
Combining \eqref{eq:appEP1} and \eqref{eq:appEP2} with \eqref{eq:appnonasymptotic6} yields
\begin{IEEEeqnarray}{lCl}
\IEEEeqnarraymulticol{3}{l}{h\Bigl(Y_k\Bigm|\bigl\{H_{3,\ell}+\xi W_{\ell}\bigr\}_{\ell=1}^{k-1}\Bigr)} \nonumber\\
\quad & \geq & \Exp\Biggl[\log\biggl(|H_{3,k}|^2\log\biggl(\frac{1}{\delta_r^2}\biggr)\delta_r\rho^2\const{A}^2\pi \nonumber\\
& & {} \qquad\qquad +e^{-\gamma}\alpha\log\biggl(\frac{1}{\delta^2}\biggr)\delta^{\alpha}\const{A}^{2\alpha}\pi+\pi e\sigma^2\biggr)\Biggr] \nonumber\\
& = & \log\pi + \log\log\frac{1}{\delta_r^2} + \log\bigl(\delta_r\rho^2\const{A}^2\bigr) \nonumber\\
& & {} + \E{\log\Biggl(|H_{3,k}|^2+\frac{e^{-\gamma}\alpha\log\bigl(\frac{1}{\delta^2}\bigr)\delta^{\alpha}\const{A}^{2\alpha}+e\sigma^2}{\log\bigl(\frac{1}{\delta_r^2}\bigr)\delta_r\rho^2\const{A}^2}\Biggr)} \nonumber\\
& = &  \log\pi + \log\log\frac{1}{\delta_r^2} + \log\bigl(\delta_r\rho^2\const{A}^2\bigr) \nonumber\\
& & {} + \log\Biggl(\frac{e^{-\gamma}\alpha\log\bigl(\frac{1}{\delta^2}\bigr)\delta^{\alpha}\const{A}^{2\alpha}+e\sigma^2}{\log\bigl(\frac{1}{\delta_r^2}\bigr)\delta_r\rho^2\const{A}^2}\Biggr) \nonumber\\
& & {} - \exp\Biggl(\frac{e^{-\gamma}\alpha\log\bigl(\frac{1}{\delta^2}\bigr)\delta^{\alpha}\const{A}^{2\alpha}+e\sigma^2}{\log\bigl(\frac{1}{\delta_r^2}\bigr)\delta_r\rho^2\const{A}^2}\Biggr)\times\nonumber\\
& & \quad\qquad \times\Ei{-\frac{e^{-\gamma}\alpha\log\bigl(\frac{1}{\delta^2}\bigr)\delta^{\alpha}\const{A}^{2\alpha}+e\sigma^2}{\log\bigl(\frac{1}{\delta_r^2}\bigr)\delta_r\rho^2\const{A}^2}} \label{eq:appnonasymptotic7}
\end{IEEEeqnarray}
where the last step follows by noting that $|H_{3,k}|^2$ has an exponential distribution with mean $1$ for which the expectation is given in \cite[Sec.~4.337]{gradshteynryzhik00}. 

Together with \eqref{eq:appbla} and \eqref{eq:appnonasymptotic5}, this yields
\begin{IEEEeqnarray}{lCl}
\IEEEeqnarraymulticol{3}{l}{I\Bigl(X_k,X_{r,k};Y_k \Bigm| \bigl\{H_{3,\ell}+\xi W_{\ell}\bigr\}_{\ell=1}^{k-1}\Bigr)}\nonumber\\
\quad & \geq & \log\Biggl(\frac{e^{-(\gamma+1)}\alpha\log\bigl(\frac{1}{\delta^2}\bigr)\delta^{\alpha}\const{A}^{2\alpha}+\sigma^2}{\delta_r\rho^2\const{A}^2}\Biggr)\nonumber\\
& & {} - \exp\Biggl(\frac{e^{-\gamma}\alpha\log\bigl(\frac{1}{\delta^2}\bigr)\delta^{\alpha}\const{A}^{2\alpha}+e\sigma^2}{\log\bigl(\frac{1}{\delta_r^2}\bigr)\delta_r\rho^2\const{A}^2}\Biggr)\times \nonumber\\
& & \qquad\quad\times \Ei{-\frac{e^{-\gamma}\alpha\log\bigl(\frac{1}{\delta^2}\bigr)\delta^{\alpha}\const{A}^{2\alpha}+e\sigma^2}{\log\bigl(\frac{1}{\delta_r^2}\bigr)\delta_r\rho^2\const{A}^2}} \nonumber\\
& & {} - \log\bigl(\eps_{r,k}^2(\xi)+\xi\bigr) \label{eq:appnonasymptotic8}
\end{IEEEeqnarray}
which, by \eqref{eq:appnoisypred}, tends to
\begin{IEEEeqnarray}{l}
 \log\Biggl(\frac{e^{-(\gamma+1)}\alpha\log\bigl(\frac{1}{\delta^2}\bigr)\delta^{\alpha}\const{A}^{2\alpha}+\sigma^2}{\delta_r\rho^2\const{A}^2}\Biggr) \nonumber\\
\quad {} -  \int_{-1/2}^{1/2}\log\biggl(F'_3(\lambda)+\frac{1}{\delta_r^2\rho^2\const{A}^{2(1-\alpha)}}+\frac{\sigma^2}{\delta_r^2\rho^2\const{A}^2}\biggr)\d\lambda \nonumber\\
 \quad {} - \exp\Biggl(\frac{e^{-\gamma}\alpha\log\bigl(\frac{1}{\delta^2}\bigr)\delta^{\alpha}\const{A}^{2\alpha}+e\sigma^2}{\log\bigl(\frac{1}{\delta_r^2}\bigr)\delta_r\rho^2\const{A}^2}\Biggr)\times\nonumber\\
\qquad\qquad \times \Ei{-\frac{e^{-\gamma}\alpha\log\bigl(\frac{1}{\delta^2}\bigr)\delta^{\alpha}\const{A}^{2\alpha}+e\sigma^2}{\log\bigl(\frac{1}{\delta_r^2}\bigr)\delta_r\rho^2\const{A}^2}} \label{eq:appnonasymptotic8_2}
\end{IEEEeqnarray}
as $k$ tends to infinity. It thus follows from \eqref{eq:appnonasymptotic2}--\eqref{eq:appnonasymptotic8_2} and the definition of the SNR that
\begin{IEEEeqnarray}{lCl}
\IEEEeqnarraymulticol{3}{l}{\lim_{n\to\infty} \frac{1}{n} I\bigl(X_1^n,X_{r,1}^n;Y_1^n\bigr)}\nonumber\\
\,\,\, & \geq & \log\Biggl(\frac{e^{-(\gamma+1)}\alpha\log\bigl(\frac{1}{\delta^2}\bigr)\delta^{\alpha}\SNR^{\alpha}\sigma^{2(\alpha-1)}+1}{\delta_r\rho^2\SNR}\Biggr) \nonumber\\
& & - \int_{-1/2}^{1/2}\log\biggl(F'_3(\lambda)+\frac{\sigma^{2(\alpha-1)}}{\delta_r^2\rho^2\SNR^{1-\alpha}}+\frac{1}{\delta_r^2\rho^2\SNR}\biggr)\d\lambda \nonumber\\
& &  {} - \exp\Biggl(\frac{e^{-\gamma}\alpha\log\bigl(\frac{1}{\delta^2}\bigr)\delta^{\alpha}\SNR^{\alpha}\sigma^{2(\alpha-1)}+e}{\log\bigl(\frac{1}{\delta_r^2}\bigr)\delta_r\rho^2\SNR}\Biggr)\times\nonumber\\
& & \qquad \times\Ei{-\frac{e^{-\gamma}\alpha\log\bigl(\frac{1}{\delta^2}\bigr)\delta^{\alpha}\SNR^{\alpha}\sigma^{2(\alpha-1)}+e}{\log\bigl(\frac{1}{\delta_r^2}\bigr)\delta_r\rho^2\SNR}} \label{eq:appnonasymptotic_L2}.\IEEEeqnarraynumspace
\end{IEEEeqnarray}
Combining \eqref{eq:appnonasymptotic_L1} and \eqref{eq:appnonasymptotic_L2} with \eqref{eq:propDF}, and maximizing over $0<\delta,\alpha,\delta_r<1$, proves Proposition~\ref{prop:nonasymptotic}.

\section{Proof of Proposition~\ref{prop:decodeforward}}
\label{app:decodeforward}
Proposition~\ref{prop:decodeforward} generalizes to channels with memory a classic result based on the decode-and-forward strategy proposed in \cite{coverelgamal79}. As in the memoryless case, we use \emph{block-Markov superposition encoding}, cf.\ \cite[Ch.~9]{kramer07}. Most steps of the proof in \cite{coverelgamal79} can be readily extended to channels with memory by defining the set of typical sequences via entropy rates rather than via entropies, cf.~\eqref{eq:appAEP}. The main difference is that for memoryless channels the events \eqref{eq:appDFevent1} and \eqref{eq:appDFevent2} below are independent of each other, whereas for channels with memory these events are dependent. Consequently, we obtain a third term on the RHS of \eqref{eq:appDFrecrate1} for which we need to show that its exponent equals to zero, cf.~\eqref{eq:appvanishmem}. For the sake of completeness, we give a detailed proof below.

\textbf{Codebook construction:} Encoding is performed in $\const{B}+1$ blocks of $n$ symbols. For each block, we generate a separate codebook. That is, we fix some distribution $P_{X,X_r}(\cdot)$ and some rate $\tilde{R}$. For every block $b=1,\ldots,\const{B}+1$ the codebook of the relay is constructed by drawing $\lfloor e^{n\tilde{R}}\rfloor$ codewords $x_{r,1}^n(v;b)$, $v=1,\dots,\lfloor e^{n\tilde{R}}\rfloor$ i.i.d.\ according to the distribution $P_{X_r}(\cdot)$. (Here, $\lfloor a\rfloor$ denotes the largest integer that is less than or equal to $a$.) As for the codebook of the transmitter, for every $v=1,\ldots,\lfloor e^{n\tilde{R}}\rfloor$ we generate $\lfloor e^{n\tilde{R}}\rfloor$ codewords $x_1^n(w,v;b)$, $w=1,\ldots,\lfloor e^{n\tilde{R}}\rfloor$ independently according to the conditional distribution $P_{X|X_r}(\cdot)$, i.e., we draw each symbol $x_{k}(w,v;b)$ according to $P_{X|X_r}\bigl(\cdot\bigm|x_{r,k}(v;b)\bigr)$.

In the proof, we assume that $P_{X,X_r}(\cdot)$ is absolutely continuous with respect to the Lebesgue measure, which implies that the random variables $(X_1^n,X_{r,1}^n,Y_{r,1}^n,Y_1^n)$ have a joint probability density function. (We shall denote the probability density function of a random variable $A$ by $f_{A}(\cdot)$.) The case where $P_{X,X_r}(\cdot)$ is not absolutely continuous with respect to the Lebesgue measure can be treated by partitioning the sample spaces of the channel inputs and outputs into a finite collection of mutually exclusive events, and by studying the resulting discrete problem following the steps below. (To this end, we need to replace the \emph{differential} entropy rates in the definition of jointly typical sequences \eqref{eq:appAEP} with entropy rates.) The result then follows by taking the supremum over all partitions, cf.\ \cite[Sec.~2.5]{gallager68}.

\textbf{Transmitter:} The message $m$ to be transmitted is divided into $\const{B}$ equally-sized blocks $m_1,\ldots,m_{\const{B}}$ of $\log\bigl(\lfloor e^{n\tilde{R}}\rfloor\bigr)$ nats each. In block $b=1,\ldots,\const{B}+1$ the transmitter sends out the codeword $x_{1}^n(m_b,m_{b-1};b)$, where we set $m_0=m_{\const{B}+1}=1$.

\textbf{Relay:} After the transmission of block $b=1,\ldots,\const{B}$ is completed, the relay has observed the sequence of outputs $y_{r,1}^n(b)$ and tries to find an $m_{r,b}$ such that, for some arbitrary $\varepsilon>0$,
\begin{IEEEeqnarray}{r}
\bigl(x_1^n(m_{r,b},\hat{m}_{r,b-1};b),x_{r,1}^n(\hat{m}_{r,b-1};b),y_{r,1}^n(b)\bigr) \qquad\IEEEeqnarraynumspace\nonumber\\
 \in \set{A}_{\varepsilon}(X_1^n,X_{r,1}^n,Y_{r,1}^n) \IEEEeqnarraynumspace\label{eq:apptypicalrelay}
\end{IEEEeqnarray}
where $\hat{m}_{r,b-1}$ denotes the relay's estimate of the message for block $b-1$, and where $\set{A}_{\varepsilon}(X_1^n,X_{r,1}^n,Y_{r,1}^n)$ denotes the \emph{set of jointly typical sequences} with respect to $P_{X_1^n,X_{r,1}^n,Y_{r,1}^n}(\cdot)$. That is
\begin{IEEEeqnarray}{lCll}
\IEEEeqnarraymulticol{4}{l}{\set{A}_{\varepsilon}(A_{1,1}^n,\ldots,A_{\tau,1}^n)}\nonumber\\
\quad & \triangleq & \Biggl\{& a_{\set{I},1}^n\in\Complex^{n|\set{I}|}, \forall \set{I}\subseteq \{1,\ldots,\tau\}\colon \nonumber\\
& & & \biggl|-\frac{1}{n}\log f_{A_{\set{I},1}^n}(a_{\set{I},1}^n)-h\bigl(\{A_{\set{I},k}\}\bigr)\biggr| < \varepsilon \Biggr\} \IEEEeqnarraynumspace\label{eq:appAEP}
\end{IEEEeqnarray}
where $a_{\set{I},1}^n$ denotes the set of sequences $a_{t,1}^n$ with $t\in\set{I}$; $|\set{I}|$ denotes the cardinality of the set $\set{I}$; and $h(\{A_{\set{I},k}\})$ denotes the entropy rate of the random processes $\{A_{t,k},\,k\in\Integers\}$, $t\in\set{I}$, i.e., we have
\begin{equation*}
h\bigl(\{A_{\set{I},k}\}\bigr) \triangleq \lim_{n\to\infty} \frac{h(A_{\set{I},1}^n)}{n}.
\end{equation*}
If one or more $m_{r,b}$ can be found satisfying \eqref{eq:apptypicalrelay}, then the relay chooses one of them, calls this choice $\hat{m}_{r,b}$, and transmits $x_{r,1}^n(\hat{m}_{r,b};b+1)$ in the subsequent block. If no such $\hat{m}_{r,b}$ is found, then the relay sets $\hat{m}_{r,b}=1$ and transmits $x_{r,1}^n(1;b+1)$ in the subsequent block.

\textbf{Receiver:} After block $b=2,\ldots,\const{B}+1$ the receiver has observed the outputs $y_{1}^n(b-1)$ and $y_{1}^n(b)$. It tries to find an $m_{b-1}$ such that
\begin{IEEEeqnarray}{r}
\bigl(x_{1}^n(m_{b-1},\hat{m}_{b-2};b-1),x_{r,1}^n(\hat{m}_{b-2};b-1),y_1^n(b-1)\bigr)\IEEEeqnarraynumspace\,\,\nonumber\\
\in \set{A}_{\varepsilon}(X_1^n,X_{r,1}^n,Y_1^n) \IEEEeqnarraynumspace\,\, \label{eq:apptypicalrec1}
\end{IEEEeqnarray}
and
\begin{IEEEeqnarray}{lCl}
\bigl(x_{r,1}^n(m_{b-1};b),y_1^n(b)\bigr) & \in & \set{A}_{\varepsilon}(X_{r,1}^n,Y_1^n) \label{eq:apptypicalrec2}
\end{IEEEeqnarray}
where $\hat{m}_{b-2}$ is the receiver's estimate of $m_{b-2}$. If one or more such $m_{b-1}$ are found, then the receiver chooses one of them and calls this choice $\hat{m}_{b-1}$. If no such $m_{b-1}$ is found, then the receiver sets $\hat{m}_{b-1}=1$.

\textbf{Analysis:} Block-Markov superposition coding is typically analyzed by upper-bounding the error probability for each block $b$ conditioned on the event $\mathscr{F}_{b-1}$ that no errors have been made up to block $b$. This approach does not work well for channels with memory. Indeed, if no errors have been made up to block $b$, then the noise and the fading in the previous blocks must be in the successful decoding regions of the relay and the receiver. Since the fading has memory, this implies that conditioning on $\mathscr{F}_{b-1}$ changes the distribution of the fading. (A similar problem occurs when analyzing the error probability of rate-splitting for multiple-access channels \cite[Sec.~II]{rimoldiurbanke96}.) We therefore analyze the error probability in a slightly different way.

For each block $b=1,\ldots,\const{B}+1$, let $\mathscr{T}_{r,b}(\hat{m}_{r,b})$ denote the event that $\hat{m}_{r,b}$ satisfies 
\begin{IEEEeqnarray}{r}
\bigl(x_{1}^n(\hat{m}_{r,b},m_{b-1};b),x_{r,1}^n(m_{b-1};b),y_{r,1}^n(b)\bigr)\qquad\IEEEeqnarraynumspace \nonumber\\
\in \set{A}_{\varepsilon}(X_1^n,X_{r,1}^n,Y_{r,1}^n) \IEEEeqnarraynumspace\label{eq:appDF_Erb+}
\end{IEEEeqnarray}
and let $\mathscr{T}^{\textnormal{c}}_{r,b}(\hat{m}_{r,b})$ denote the event that $\hat{m}_{r,b}$ does not satisfy \eqref{eq:appDF_Erb+}.  Similarly, let $\mathscr{T}_b(\hat{m}_{b-1})$ denote the event that $\hat{m}_{b-1}$ satisfies
\begin{IEEEeqnarray}{r}
\bigl(x_1^n(\hat{m}_{b-1},m_{b-2};b-1),x_{r,1}^n(m_{b-2};b-1),y_1^n(b-1)\bigr)\IEEEeqnarraynumspace\nonumber\\
 \in \set{A}_{\varepsilon}(X_1^n,X_{r,1}^n,Y_1^n) \IEEEeqnarraynumspace\label{eq:appDFevent1}
\end{IEEEeqnarray}
and
\begin{IEEEeqnarray}{lCl}
\bigl(x_{r,1}^n(\hat{m}_{b-1};b),y_1^n(b)\bigr) & \in & \set{A}_{\varepsilon}(X_{r,1}^n,Y_1^n)\label{eq:appDFevent2}
\end{IEEEeqnarray}
and let $\mathscr{T}_b^{\textnormal{c}}(\hat{m}_{b-1})$ denote the event that $\hat{m}_{b-1}$ does not satisfy \eqref{eq:appDFevent1} and \eqref{eq:appDFevent2}. The event that either the relay 
or the receiver makes an error in at least one of the blocks is a subset of the union of events
\begin{IEEEeqnarray*}{l}
\Biggl(\bigcup_{b=1}^{\const{B}} \mathscr{T}^{\textnormal{c}}_{r,b}(m_b) \cup \bigcup_{\hat{m}_{r,b}\neq m_b} \mathscr{T}_{r,b}(\hat{m}_{r,b})\Biggr) \\
\qquad\quad \cup  \Biggl(\bigcup_{b=2}^{\const{B}+1} \mathscr{T}^{\textnormal{c}}_b(m_{b-1}) \bigcup_{\hat{m}_{b-1}\neq m_{b-1}} \mathscr{T}_b(\hat{m}_{b-1})\Biggr).
\end{IEEEeqnarray*}
It thus follows from the union bound that the probability of error is upper-bounded by
\begin{IEEEeqnarray}{lCl}
\IEEEeqnarraymulticol{3}{l}{\Prob(\text{error})} \nonumber\\
 & \leq &\Prob\Biggl( \mathscr{T}^{\textnormal{c}}_{r,1}(m_1) \cup \bigcup_{\hat{m}_{r,1}\neq m_1} \mathscr{T}_{r,1}(\hat{m}_{r,1})\Biggr)  \nonumber\\
& & {} + \sum_{b=2}^{\const{B}}\Biggl[\Prob\bigl(\mathscr{T}^{\textnormal{c}}_{r,b}(m_b)\cup\mathscr{T}^{\textnormal{c}}_b(m_{b-1})\bigr)\nonumber\\
& & \quad {} + \Prob\Biggl(\bigcup_{\hat{m}_{r,b}\neq m_b} \mathscr{T}_{r,b}(\hat{m}_{r,b})\cup\bigcup_{\hat{m}_{b-1}\neq m_{b-1}} \mathscr{T}_b(\hat{m}_{b-1})\Biggr)  \Biggr] \nonumber\\
& & {} + \Prob\Biggl(\mathscr{T}^{\textnormal{c}}_{\const{B}+1}(m_{\const{B}}) \cup \bigcup_{\hat{m}_{\const{B}}\neq m_{\const{B}}} \mathscr{T}_{\const{B}+1}(\hat{m}_{\const{B}})\Biggr). \label{eq:app_T0}
\end{IEEEeqnarray}

We next upper-bound the error probability for each block $b$. The overall probability of error is then upper-bounded by $(\const{B}+1)$ times the maximum error probability of each block. Consequently, if for each block the error probability tends to zero as $n$ tends to infinity, then so does the overall probability of error. 

In order to upper-bound 
\begin{equation}
\label{eq:app_notyp}
\Prob\bigl(\mathscr{T}^{\textnormal{c}}_{r,b}(m_b)\cup\mathscr{T}^{\textnormal{c}}_b(m_{b-1})\bigr)
\end{equation}
note that for a given $(m_b,m_{b-1})$ the process \mbox{$\{(X_k,X_{r,k}),\,k\in\Integers\}$} is i.i.d.\ and jointly independent of the stationary and ergodic, complex Gaussian fading processes \mbox{$\{H_{\ell,k},\,k\in\Integers\}$}, \mbox{$\ell=1,2,3$} and of the i.i.d.\ Gaussian noise processes \mbox{$\{Z_{r,k},\,k\in\Integers\}$} and $\{Z_k,\,k\in\Integers\}$. This implies that the process $\{(X_k,X_{r,k},Y_{r,k},Y_k),\,k\in\Integers\}$ is jointly stationary and ergodic, and the Shannon-McMillan-Breiman Theorem \cite[Th.~2]{algoetcover88} yields \eqref{eq:appDFtypical1}--\eqref{eq:appDFtypical3}, shown at the top of this page, which imply that \eqref{eq:app_notyp} tends to zero as $n$ tends to infinity. 
\begin{figure*}[t!]
\begin{IEEEeqnarray}{rCll}
\lim_{n\to\infty} \Prob\Bigl(\bigl(X_1^n(m_{b},m_{b-1};b),X_{r,1}^n(m_{b-1};b),Y_{r,1}^n(b)\bigr)\in\set{A}_{\varepsilon}(X_1^n,X_{r,1}^n,Y_{r,1}^n)\Bigr) & = & 1, \quad & \varepsilon>0 \label{eq:appDFtypical1}\\
\lim_{n\to\infty} \Prob\Bigl(\bigl(X_1^n(m_{b-1},m_{b-2};b-1),X_{r,1}^n(m_{b-2};b-1),Y_1^n(b-1)\bigr)\in\set{A}_{\varepsilon}(X_1^n,X_{r,1}^n,Y_1^n)\Bigr) & = & 1, \quad & \varepsilon>0 \IEEEeqnarraynumspace \label{eq:appDFtypical2}\\
\lim_{n\to\infty} \Prob\Bigl(\bigl(X_{r,1}^n(m_{b-1};b),Y_1^n(b)\bigr)\in\set{A}_{\varepsilon}(X_{r,1}^n,Y_1^n)\Bigr) & = & 1, \quad & \varepsilon>0.\label{eq:appDFtypical3}
\end{IEEEeqnarray}
 \rule{\textwidth}{0.5pt}
\end{figure*}

We continue by upper-bounding
\begin{IEEEeqnarray}{lCl}
\IEEEeqnarraymulticol{3}{l}{\Prob\Biggl(\bigcup_{\hat{m}_{r,b}\neq m_b} \mathscr{T}_{r,b}(\hat{m}_{r,b})\cup\bigcup_{\hat{m}_{b-1}\neq m_{b-1}} \mathscr{T}_b(\hat{m}_{b-1})\Biggr)}\nonumber\\
 & \leq & \sum_{\hat{m}_{r,b}\neq m_b} \Prob\bigl(\mathscr{T}_{r,b}(\hat{m}_{r,b})\bigr) + \sum_{\hat{m}_{b-1}\neq m_{b-1}} \Prob\bigl(\mathscr{T}_b(\hat{m}_{b-1})\bigr) \IEEEeqnarraynumspace\label{eq:app_T1}
\end{IEEEeqnarray}
using the union bound. To analyze the summands in the first sum on the RHS of \eqref{eq:app_T1}, note that for $\hat{m}_{r,b}\neq m_b$ the triple \[\bigl(X_1^n(\hat{m}_{r,b},m_{r,b-1};b),X_{r,1}^n(m_{b-1};b),Y_{r,1}^n(b)\bigr)\] is distributed according to \[P_{X_{r,1}^n}(\cdot) P_{X_1^n|X_{r,1}^n}(\cdot) P_{Y_{r,1}^n|X_{r,1}^n}(\cdot).\] Generalizing \cite[Th.~14.2.3]{coverthomas91} to channels with memory\footnote{To this end, we need to replace the entropies in the proof of \cite[Th.~14.2.3]{coverthomas91} by the corresponding differential entropy rates.} yields that, for every $\hat{m}_{r,b}\neq m_b$, we have
\begin{IEEEeqnarray}{lCl}
\IEEEeqnarraymulticol{3}{l}{\Prob\bigl(\mathscr{T}_{r,b}(\hat{m}_{r,b})\bigr)} \nonumber\\
\quad & \leq & \exp\Biggl(-n \biggl(\lim_{\eta\to\infty}\frac{1}{\eta} I\bigl(X_1^{\eta};Y_{r,1}^{\eta} \bigm| X_{r,1}^{\eta}\bigr)-6 \varepsilon\biggr) \Biggr). \IEEEeqnarraynumspace
\end{IEEEeqnarray}
Since there are $\lfloor e^{n\tilde{R}}\rfloor-1$ different values for $\hat{m}_{r,b}\neq m_b$, this yields
\begin{IEEEeqnarray}{lCl}
\IEEEeqnarraymulticol{3}{l}{\sum_{\hat{m}_{r,b}\neq m_b} \Prob\bigl(\mathscr{T}_{r,b}(\hat{m}_{r,b})\bigr)} \nonumber\\
\quad & \leq & \exp\Biggl(n \biggl(\tilde{R} - \lim_{\eta\to\infty}\frac{1}{\eta} I\bigl(X_1^{\eta};Y_{r,1}^{\eta} \bigm| X_{r,1}^{\eta}\bigr)+6 \varepsilon\biggr) \Biggr) \IEEEeqnarraynumspace\label{eq:appDFrelayrate}
\end{IEEEeqnarray}
which implies that if
\begin{equation*}
\tilde{R} < \lim_{n\to\infty}\frac{1}{n} I\bigl(X_1^n;Y_{r,1}^n \bigm| X_{r,1}^n\bigr)-6 \varepsilon
\end{equation*}
then the first sum on the RHS of \eqref{eq:app_T1} vanishes as $n$ tends to infinity.

To analyze the summands in the second sum on the RHS of \eqref{eq:app_T1}, note that for $\hat{m}_{b-1}\neq m_{b-1}$ the tuple
\begin{IEEEeqnarray*}{r}
\bigl(X_1^n(\hat{m}_{b-1},m_{b-2};b-1),X_{r,1}^n(m_{b-2};b-1),\qquad\\
Y_1^n(b-1),X_{r,1}^n(\hat{m}_{b-1};b),Y_1^n(b)\bigr)
\end{IEEEeqnarray*}
is distributed according to
\begin{IEEEeqnarray*}{r}
P_{X_{r,1}^n(b-1)}(\cdot) P_{X_1^n(b-1)|X_{r,1}^n(b-1)}(\cdot) P_{Y_{1}^n(b-1)|X_{r,1}^n(b-1)}(\cdot)\times\quad\,\,\\ 
\times P_{X_{r,1}^n(b)}(\cdot) P_{Y_1^n(b)|X_{r,1}^n(b-1),Y_1^n(b-1)}(\cdot)
\end{IEEEeqnarray*}
where the arguments after the random variables indicate whether the vector belongs to block $b$ or $(b-1)$. Extending \cite[Ths.~14.2.1 \& 14.2.3]{coverthomas91} to the above channel, we obtain for every $\hat{m}_{b-1}\neq m_{b-1}$
\begin{IEEEeqnarray}{lCl}
\IEEEeqnarraymulticol{3}{l}{\Prob\bigl(\mathscr{T}_b(\hat{m}_{b-1})\bigr)} \nonumber\\
\quad & \leq & e^{-n\Bigl(\lim_{\eta\to\infty}\frac{1}{\eta} I\bigl(X_1^{\eta}(b-1);Y_1^{\eta}(b-1)\bigm| X_{r,1}^{\eta}(b-1)\bigr)-6\varepsilon\Bigr)}\nonumber\\
& & {} \times e^{-n\Bigl(\lim_{\eta\to\infty} \frac{1}{\eta} I\bigl(X_{r,1}^{\eta}(b);Y_1^{\eta}(b)\bigr)-4\varepsilon\Bigr)}\nonumber\\
& & {} \times e^{n\Bigl(\lim_{\eta\to\infty}\frac{1}{\eta} I\bigl(Y_1^{\eta}(b);X_{r,1}^{\eta}(b-1),Y_1^{\eta}(b-1)\bigr) \Bigr)}. \label{eq:appDFrecrate1}
\end{IEEEeqnarray}
Since the codebook construction does not depend on the block $b$, and since the channel is stationary, it follows that
\begin{IEEEeqnarray}{lCl}
\IEEEeqnarraymulticol{3}{l}{\lim_{n\to\infty}\frac{1}{n} I\bigl(X_1^n(b-1);Y_1^n(b-1)\bigm| X_{r,1}^n(b-1)\bigr)} \nonumber\\
\quad & = & \lim_{n\to\infty}\frac{1}{n} I\bigl(X_1^n;Y_1^n\bigm| X_{r,1}^n\bigr)
\end{IEEEeqnarray}
and
\begin{equation}
\lim_{n\to\infty} \frac{1}{n} I\bigl(X_{r,1}^n(b);Y_1^n(b)\bigr) = \lim_{n\to\infty} \frac{1}{n} I\bigl(X_{r,1}^n;Y_1^n\bigr).
\end{equation}
We next show that
\begin{equation}
\label{eq:appvanishmem}
\lim_{n\to\infty}\frac{1}{n} I\bigl(Y_1^n(b);X_{r,1}^n(b-1),Y_1^n(b-1)\bigr) = 0.
\end{equation}
We first note that by the stationarity of \[\bigl(Y_1^n(b),X_{r,1}^n(b-1),Y_1^n(b-1)\bigr)\] we have
\begin{IEEEeqnarray}{lCl}
\IEEEeqnarraymulticol{3}{l}{I\bigl(Y_1^n(b);X_{r,1}^n(b-1),Y_1^n(b-1)\bigr)} \nonumber\\
\qquad\qquad & = & I\bigl(Y_1^n;X_{r,-n+1}^0,Y_{-n+1}^0\bigr).\label{eq:appvanishproof1}
\end{IEEEeqnarray}
Using the Data Processing Inequality on the Markov chain
\begin{IEEEeqnarray*}{lCl}
Y_1^n & \markov & \bigl(H_{2,1}^n,H_{3,1}^n\bigr) \\
& \markov & \bigl(H_{2,-n+1}^0,H_{3,-n+1}^0\bigr) \\
& \markov & \bigl(X_{r,-n+1}^0,Y_{-n+1}^0\bigr)
\end{IEEEeqnarray*}
the expression \eqref{eq:appvanishproof1} can be upper-bounded by
\begin{IEEEeqnarray}{lCl}
\IEEEeqnarraymulticol{3}{l}{I\bigl(Y_1^n;X_{r,-n+1}^0,Y_{-n+1}^0\bigr)} \nonumber\\ 
\quad & \leq & I\bigl(H_{2,1}^n,H_{3,1}^n;H_{2,-n+1}^0,H_{3,-n+1}^0\bigr).
\end{IEEEeqnarray}
Since the processes $\{H_{2,k},\,k\in\Integers\}$ and $\{H_{3,k},\,k\in\Integers\}$ are independent, it follows that
\begin{IEEEeqnarray}{lCl}
\IEEEeqnarraymulticol{3}{l}{I\bigl(H_{2,1}^n,H_{3,1}^n;H_{2,-n+1}^0,H_{3,-n+1}^0\bigr)}\nonumber\\
\quad & = & I\bigl(H_{2,1}^n;H_{2,-n+1}^0\bigr) + I\bigl(H_{3,1}^n;H_{3,-n+1}^0\bigr) \nonumber\\
& \leq & I\bigl(H_{2,1}^n;H_{2,-\infty}^0\bigr) + I\bigl(H_{3,1}^n;H_{3,-\infty}^0\bigr)
\end{IEEEeqnarray}
where the inequality follows because adding observations does not decrease mutual information.

Using the chain rule for mutual information and Ces\'aro's mean, we obtain
\begin{IEEEeqnarray}{lCl}
\IEEEeqnarraymulticol{3}{l}{\lim_{n\to\infty} \frac{1}{n} I\bigl(H_{2,1}^n;H_{2,-\infty}^0\bigr)} \nonumber\\
\quad & =  & \lim_{k\to\infty} I\bigl(H_{2,k};H_{2,-\infty}^0\bigm| H_{2,1}^{k-1}\bigr) \nonumber\\
& = & \lim_{k\to\infty} h\bigl(H_{2,k}\bigm|H_{2,1}^{k-1}\bigr) - \lim_{k\to\infty} h\bigl(H_{2,k}\bigm| H_{2,-\infty}^{k-1}\bigr) \nonumber\\
& = & h\bigl(\{H_{2,k}\}\bigr) - h\bigl(\{H_{2,k}\}\bigr) \nonumber\\
& = & 0
\end{IEEEeqnarray}
where the third step follows from the stationarity of $\{H_{2,k},\,k\in\Integers\}$ \cite[Th.~4.2.1]{coverthomas91}. In the same way, it can be shown that
\begin{equation}
\label{eq:appvanishproof2}
\lim_{n\to\infty} \frac{1}{n} I\bigl(H_{3,1}^n;H_{3,-\infty}^0\bigr) = 0.
\end{equation}
Combining \eqref{eq:appvanishproof1}--\eqref{eq:appvanishproof2} proves \eqref{eq:appvanishmem}. We thus obtain from \eqref{eq:appDFrecrate1}--\eqref{eq:appvanishmem} that
\begin{IEEEeqnarray}{lCl}
\IEEEeqnarraymulticol{3}{l}{\Prob\bigl(\mathscr{T}_b(\hat{m}_{b-1})\bigr)} \nonumber\\
\quad & \leq & e^{-n\bigl( \lim_{\eta\to\infty}\frac{1}{\eta} I(X_1^{\eta};Y_1^{\eta} | X_{r,1}^{\eta}) + \lim_{\eta\to\infty} \frac{1}{\eta} I(X_{r,1}^{\eta};Y_1^{\eta})-10\varepsilon\bigr)}.\nonumber\\ \label{eq:appDFrecrate2}
\end{IEEEeqnarray}
Since there are $\lfloor e^{n\tilde{R}}\rfloor-1$ different values for $\hat{m}_{b-1}\neq m_{b-1}$ this yields
\begin{IEEEeqnarray}{lCl}
\IEEEeqnarraymulticol{3}{l}{\sum_{\hat{m}_{b-1}\neq m_{b-1}} \Prob\bigl(\mathscr{T}_b(\hat{m}_{b-1})\bigr)}\nonumber\\
\quad & \leq & e^{n\bigl(\tilde{R} - \lim_{\eta\to\infty}\frac{1}{\eta} I(X_1^{\eta};Y_1^{\eta} | X_{r,1}^{\eta})+\lim_{\eta\to\infty} \frac{1}{\eta} I(X_{r,1}^{\eta};Y_1^{\eta})+10\varepsilon\bigr)} \nonumber\\
& = & e^{n\bigl(\tilde{R} - \lim_{\eta\to\infty}\frac{1}{\eta} I(X_1^{\eta},X_{r,1}^{\eta};Y_1^{\eta})+10\varepsilon\bigr)}\label{eq:appDFrecrate3}
\end{IEEEeqnarray}
which implies that the second sum on the RHS of \eqref{eq:app_T1} vanishes as $n$ tends to infinity, provided that
\begin{equation*}
\tilde{R} < \lim_{n\to\infty}\frac{1}{n} I\bigl(X_1^n,X_{r,1}^n;Y_1^n\bigr)-10\varepsilon.
\end{equation*}
Since $\varepsilon>0$ is arbitrary, it follows from \eqref{eq:app_T0},  \eqref{eq:appDFtypical1}--\eqref{eq:app_T1}, \eqref{eq:appDFrelayrate}, and \eqref{eq:appDFrecrate3} that for every block $b$ any rate satisfying
\begin{IEEEeqnarray}{lCll}
\tilde{R} & < & \min\Biggl\{& \lim_{n\to\infty}\frac{1}{n} I\bigl(X_1^n;Y_{r,1}^n \bigm| X_{r,1}^n\bigr), \nonumber\\
& & & \quad\,\lim_{n\to\infty}\frac{1}{n} I\bigl(X_1^n,X_{r,1}^n;Y_1^n\bigr)\Biggr\} \label{eq:appDFthis_is_the_rate}
\end{IEEEeqnarray}
is achievable. Consequently, the complete message $m$ can be transmitted over $\const{B}+1$ blocks with an overall rate
\begin{equation}
R = \frac{\const{B}\tilde{R}}{\const{B}+1}.
\end{equation}
By letting $\const{B}$ tend to infinity, it follows that, for every i.i.d.\ process $\{(X_k,X_{r,k}),\,k\in\Integers\}$, we can achieve the rate \eqref{eq:appDFthis_is_the_rate}, thus proving Proposition~\ref{prop:decodeforward}.

\section{The limit of $\Delta_2(\SNR,\kappa)$}
\label{app:vareps2}
In the following we show that $\Delta_2(\SNR,\kappa)$ tends to zero as $\SNR$ tends to infinity. The proof follows along the same lines as the proof in \cite[Appendix IX]{lapidothmoser03_3}.

We first note that $\Delta_2(\SNR,\kappa) \geq 0$. It thus suffices to show
that $\varlimsup_{\kappa\to\infty}\Delta_2(\SNR,\kappa) \leq 0$. We have
\begin{IEEEeqnarray}{lCl}
  \IEEEeqnarraymulticol{3}{l}{\Delta_2(\SNR,\kappa)} \nonumber\\
  \quad &  = &
  I\bigl(X_k,X_{r,k};H_{3,k-\kappa}^{k-1}\bigm|Y_{k-\kappa}^k,X_{k-\kappa}^{k-1},X_{r,k-\kappa}^{k-1}\bigr)\nonumber\\
  & = & h\bigl(H_{3,k-\kappa}^{k-1}\bigm|Y_{k-\kappa}^k,X_{k-\kappa}^{k-1},X_{r,k-\kappa}^{k-1})\nonumber\\
  & & {} - h\bigl(H_{3,k-\kappa}^{k-1}\bigm|Y_{k-\kappa}^k,X_{k-\kappa}^{k},X_{r,k-\kappa}^{k}\bigr)\nonumber\\
  & \leq & h\bigl(H_{3,k-\kappa}^{k-1}\bigm|Y_{k-\kappa}^{k-1},X_{k-\kappa}^{k-1},X_{r,k-\kappa}^{k-1}\bigr) \nonumber\\
  & & {} - h\bigl(H_{3,k-\kappa}^{k-1}\bigm|Y_{k-\kappa}^k,X_{k-\kappa}^{k},X_{r,k-\kappa}^{k},H_{3,k}\bigr)\nonumber\\
  & = &  h\bigl(H_{3,k-\kappa}^{k-1}\bigm|Y_{k-\kappa}^{k-1},X_{k-\kappa}^{k-1},X_{r,k-\kappa}^{k-1}\bigr) \nonumber\\
  & & {} -h\bigl(H_{3,k-\kappa}^{k-1}\bigm|Y_{k-\kappa}^{k-1},X_{k-\kappa}^{k-1},X_{r,k-\kappa}^{k-1},H_{3,k}\bigr)\nonumber\\
  & = & I\bigl(H_{3,k-\kappa}^{k-1};H_{3,k}\bigm|Y_{k-\kappa}^{k-1},X_{k-\kappa}^{k-1},X_{r,k-\kappa}^{k-1}\bigr) \label{eq:156}
  \end{IEEEeqnarray}
where the inequality follows because conditioning cannot increase entropy; the subsequent equality follows because, conditioned on $\bigl(Y_{k-\kappa}^{k-1},X_{k-\kappa}^{k-1},X_{r,k-\kappa}^{k-1},H_{3,k}\bigr)$, the fading coefficients $H_{3,k-\kappa}^{k-1}$ are independent of $\bigl(Y_k,X_k,X_{r,k}\bigr)$.
We thus have
\begin{IEEEeqnarray}{lCl}
 \IEEEeqnarraymulticol{3}{l}{\Delta_2(\SNR,\kappa)} \nonumber\\
\quad & \leq &
  h\bigl(H_{3,k}\bigm|Y_{k-\kappa}^{k-1},X_{k-\kappa}^{k-1},X_{r,k-\kappa}^{k-1}\bigr) \nonumber\\
  & & {} - h\bigl(H_{3,k}\bigm|Y_{k-\kappa}^{k-1},X_{k-\kappa}^{k-1},X_{r,k-\kappa}^{k-1},H_{3,k-\kappa}^{k-1}\bigr)\nonumber\\
  & = & h\bigl(H_{3,k}\bigm|Y_{k-\kappa}^{k-1},X_{k-\kappa}^{k-1},X_{r,k-\kappa}^{k-1}\bigr) - h\bigl(H_{3,k}\bigm|H_{3,k-\kappa}^{k-1}\bigr)\nonumber\\
  & = & h\Biggl(H_{3,k}\Biggm|\biggl\{\frac{Y_{\ell}}{X_{r,\ell}}\biggr\}_{\ell=k-\kappa}^{k-1},X_{k-\kappa}^{k-1},X_{r,k-\kappa}^{k-1}\Biggr) \nonumber\\
  & & {} - h\bigl(H_{3,k}\bigm|H_{3,k-\kappa}^{k-1}\bigr)\nonumber\\
  & \leq & h\Bigl(H_{3,k}\Bigm|\bigl\{H_{3,\ell}+\zeta\, W_{\ell}\bigr\}_{\ell=k-\kappa}^{k-1}\Bigr) \nonumber\\
  & & {} -h\bigl(H_{3,k}\bigm| H_{3,k-\kappa}^{k-1}\bigr)\label{eq:appvareps2_1}
\end{IEEEeqnarray}
where $\{W_{k},\,k\in\Integers\}$ is a sequence of i.i.d., zero-mean, unit-variance, circularly-symmetric, complex Gaussian random variables, and where
\begin{equation*}
  \zeta \triangleq \frac{\const{A}^{2\alpha}}{\const{A}_r^{2\beta}} + \frac{\sigma^2}{\const{A}_r^{2\beta}}.
\end{equation*}
The second step in \eqref{eq:appvareps2_1} follows because, conditioned on $H_{3,k-\kappa}^{k-1}$, the present fading $H_{3,k}$ is independent of $\bigl(X_{k-\kappa}^{k-1},X_{r,k-\kappa}^{k-1},Y_{k-\kappa}^{k-1}\bigr)$; the last step in \eqref{eq:appvareps2_1} follows because the first differential entropy is maximized for $|x_{\ell}|^2=\const{A}^{2\alpha}$ and $|x_{r,\ell}|^2=\const{A}_r^{2\beta}$, in which case
\begin{equation*}
\biggl\{\frac{Y_{\ell}}{X_{r,\ell}}\biggr\}_{\ell=k-\kappa}^{k-1}= \biggl\{H_{2,\ell}\frac{X_{\ell}}{X_{r,\ell}}+H_{3,\ell}+\frac{Z_{\ell}}{X_{r,\ell}}\biggr\}_{\ell=k-\kappa}^{k-1}
\end{equation*}
has the same law as $\bigl\{H_{3,\ell}+\zeta\, W_{\ell}\bigr\}_{\ell=k-\kappa}^{k-1}$. Noting that
\begin{IEEEeqnarray}{lCl}
\IEEEeqnarraymulticol{3}{l}{h\Bigl(H_{3,k}\Bigm|\bigl\{H_{3,\ell}+\zeta\, W_{\ell}\bigr\}_{\ell=k-\kappa}^{k-1}\Bigr)-h\bigl(H_{3,k}\bigm| H_{3,k-\kappa}^{k-1}\bigr)}\nonumber\\
\quad & = & h\Bigl(H_{3,k},\bigl\{H_{3,\ell}+\zeta\, W_{\ell}\bigr\}_{\ell=k-\kappa}^{k-1}\Bigr)-h\bigl(H_{3,k-\kappa}^{k}\bigr) \nonumber\\
& & {} - I\Bigl(\bigl\{H_{3,\ell}+\zeta\, W_{\ell}\bigr\}_{\ell=k-\kappa}^{k-1};W_{k-\kappa}^{k-1}\Bigr) \nonumber\\
& \leq & h\Bigl(H_{3,k},\bigl\{H_{3,\ell}+\zeta\, W_{\ell}\bigr\}_{\ell=k-\kappa}^{k-1}\Bigr)-h\bigl(H_{3,k-\kappa}^{k}\bigr) \IEEEeqnarraynumspace
\end{IEEEeqnarray}
we obtain
\begin{IEEEeqnarray}{lCl}
\IEEEeqnarraymulticol{3}{l}{\Delta_2(\SNR,\kappa)} \nonumber\\
\quad  &\leq & h\Bigl(H_{3,k},\bigl\{H_{3,\ell}+\zeta\, W_{\ell}\bigr\}_{\ell=k-\kappa}^{k-1}\Bigr) - h\bigl(H_{3,k-\kappa}^{k}\bigr).\IEEEeqnarraynumspace
\end{IEEEeqnarray}
The claim now follows by \cite[Lemma 6.11]{lapidothmoser03_3}, which states that if $\vect{H} \in \Complex^{\nu}$ is a random vector of finite Frobenius norm and finite differential entropy, and if $\vect{W} \in \Complex^{\nu}$ is a Gaussian random vector that is independent of $\vect{H}$, then
\begin{equation*}
  \lim_{\sigma^2 \to 0} \left\{h(\vect{H}+\sigma^2\vect{W})-h(\vect{H})\right\} = 0.
\end{equation*}

\section*{Acknowledgment}
The authors would like to thank Amos Lapidoth, Alfonso Martinez, the Associate Editor Elza Erkip, and the anonymous referees for their constructive comments.


\begin{thebibliography}{10}
\providecommand{\url}[1]{#1}
\csname url@samestyle\endcsname
\providecommand{\newblock}{\relax}
\providecommand{\bibinfo}[2]{#2}
\providecommand{\BIBentrySTDinterwordspacing}{\spaceskip=0pt\relax}
\providecommand{\BIBentryALTinterwordstretchfactor}{4}
\providecommand{\BIBentryALTinterwordspacing}{\spaceskip=\fontdimen2\font plus
\BIBentryALTinterwordstretchfactor\fontdimen3\font minus
  \fontdimen4\font\relax}
\providecommand{\BIBforeignlanguage}[2]{{%
\expandafter\ifx\csname l@#1\endcsname\relax
\typeout{** WARNING: IEEEtran.bst: No hyphenation pattern has been}%
\typeout{** loaded for the language `#1'. Using the pattern for}%
\typeout{** the default language instead.}%
\else
\language=\csname l@#1\endcsname
\fi
#2}}
\providecommand{\BIBdecl}{\relax}
\BIBdecl

\bibitem{kramergastpargupta05}
G.~Kramer, M.~Gastpar, and P.~Gupta, ``Cooperative strategies and capacity
  theorems for relay networks,'' \emph{IEEE Trans. Inf. Theory}, vol.~51,
  no.~9, pp. 3037--3063, Sept. 2005.

\bibitem{wangzhanghostmadsen05}
B.~Wang, J.~Zhang, and A.~H{\o}st-Madsen, ``On the capacity of {MIMO} relay
  channels,'' \emph{IEEE Trans. Inf. Theory}, vol.~51, no.~1, pp. 29--43,
  Jan. 2005.

\bibitem{lapidothmoser03_3}
A.~Lapidoth and S.~M. Moser, ``Capacity bounds via duality with applications to
  multiple-antenna systems on flat fading channels,'' \emph{IEEE Trans. Inf.
  Theory}, vol.~49, no.~10, pp. 2426--2467, Oct. 2003.

\bibitem{ericson70}
T.~H.~E. Ericson, ``A {Gaussian} channel with slow fading,'' \emph{IEEE Trans.
  Inf. Theory}, vol.~16, no.~3, pp. 353--355, May 1970.

\bibitem{lapidoth05}
A.~Lapidoth, ``On the asymptotic capacity of stationary {G}aussian fading
  channels,'' \emph{IEEE Trans. Inf. Theory}, vol.~51, no.~2, pp. 437--446,
  Feb. 2005.

\bibitem{marzettahochwald99}
T.~L. Marzetta and B.~M. Hochwald, ``Capacity of a mobile multiple-antenna
  communication link in {R}ayleigh flat fading,'' \emph{IEEE Trans. Inf.
  Theory}, vol.~45, no.~1, pp. 139--157, Jan. 1999.

\bibitem{zhengtse02}
L.~Zheng and D.~N.~C. Tse, ``Communicating on the {Grassmann} manifold: a
  geometric approach to the non coherent multiple-antenna channel,'' \emph{IEEE
  Trans. Inf. Theory}, vol.~48, no.~2, pp. 359--383, Feb. 2002.

\bibitem{liangveeravalli04}
Y.~Liang and V.~V. Veeravalli, ``Capacity of noncoherent time-selective
  Rayleigh-fading channels,'' \emph{IEEE Trans. Inf. Theory}, vol.~50,
  no.~12, pp. 3095--3110, Dec. 2004.

\bibitem{doob90}
J.~Doob, \emph{Stochastic Processes}.\hskip 1em plus 0.5em minus 0.4em\relax
  John Wiley \& Sons, 1990.

\bibitem{kochkramer05}
T.~Koch and G.~Kramer, ``On the pre-log of {G}aussian fading relay channels,''
  in \emph{Proc. IEEE Int. Symp. Inf. Theory}, Adelaide, Australia,
  Sept. 4--9, 2005, pp. 951--955.

\bibitem{priestley81}
M.~B. Priestley, \emph{Spectral Analysis and Time Series}, ser. Probability and
  Mathematical Statistics.\hskip 1em plus 0.5em minus 0.4em\relax Academic
  Press, 1981.

\bibitem{jakes75}
W.~C. Jakes, \emph{Microwave Mobile Communications}.\hskip 1em plus 0.5em minus
  0.4em\relax John Wiley \& Sons, 1975.

\bibitem{etkintse06}
R.~Etkin and D.~Tse, ``Degrees of freedom in some underspread {MIMO} fading
  channels,'' \emph{IEEE Trans. Inf. Theory}, vol.~52, no.~4, pp.
  1576--1608, Apr. 2006.

\bibitem{lapidothmoser06}
A.~Lapidoth and S.~M. Moser, ``The fading number of single-input
  multiple-output fading channels with memory,'' \emph{IEEE Trans. Inf.
  Theory}, vol.~52, no.~2, pp. 437--453, Feb. 2006.

\bibitem{lapidoth03_2}
A.~Lapidoth, ``On the high {SNR} capacity of stationary {G}aussian fading
  channels,'' in \emph{Proc. 41st Allerton Conf. Comm., Contr. and Comp.},
  Allerton, Monticello, Il, Oct. 1--3, 2003, pp. 410--419.

\bibitem{kochlapidoth05_1}
T.~Koch and A.~Lapidoth, ``Degrees of freedom in non-coherent stationary {MIMO}
  fading channels,'' in \emph{Proc. Winter School Cod. and Inf. Theory},
  Bratislava, Slovakia, Feb. 20--25{,} 2005, pp. 91--97.

\bibitem{coverthomas91}
T.~M. Cover and J.~A. Thomas, \emph{Elements of Information Theory},
  1st~ed.\hskip 1em plus 0.5em minus 0.4em\relax John Wiley \& Sons, 1991.

\bibitem{kochlapidoth05_3}
T.~Koch and A.~Lapidoth, ``The fading number and degrees of freedom in
  non-coherent {MIMO} fading channels: a peace pipe,'' in \emph{Proc. IEEE Int.
  Symp. Inf. Theory}, Adelaide, Australia, Sept. 4--9, 2005, pp.
  661--665.

\bibitem{koch04}
T.~Koch, ``On the asymptotic capacity of multiple-input single-output fading
  channels with memory,'' Master's thesis, Signal and Inform. Proc. Lab., ETH
  Zurich, Switzerland, Apr. 2004, supervised by Prof. Dr. Amos Lapidoth.

\bibitem{coverelgamal79}
T.~M. Cover and A.~A. El~Gamal, ``Capacity theorems for the relay channel,''
  \emph{IEEE Trans. Inf. Theory}, vol.~25, no.~5, pp. 572--584, Sept. 1979.

\bibitem{avestimehrdiggavitse11}
A.~S. Avestimehr, S.~N. Diggavi, and D.~N.~C. Tse, ``Wireless network
  information flow: A deterministic approach,'' \emph{IEEE Trans. Inf.
  Theory}, vol.~57, no.~4, pp. 1872--1905, Apr. 2011.

\bibitem{koch09}
T.~Koch, ``On heating up and fading in communication channels,'' Ph.D.
  dissertation, Swiss Federal Institute of Technology, Zurich, May 2009,
  {D}iss. ETH No. 18339.

\bibitem{sethuramanhajeknarayanan05}
V.~Sethuraman, B.~Hajek, and K.~Narayanan, ``Capacity bounds for noncoherent
  fading channels with a peak constraint,'' in \emph{Proc. IEEE Int. Symp. Inf. Theory}, Adelaide, Australia, Sept. 4--9, 2005, pp. 515--519.

\bibitem{sethuramanwanghajeklapidoth09}
V.~Sethuraman, L.~Wang, B.~Hajek, and A.~Lapidoth, ``Low-{SNR} capacity of
  noncoherent fading channels,'' \emph{IEEE Trans. Inf. Theory}, vol.~55,
  no.~4, pp. 1555--1574, Apr. 2009.

\bibitem{lapidothshamai02}
A.~Lapidoth and S.~Shamai~(Shitz), ``Fading channels: how perfect need
  {`}perfect side-information{'} be{?}'' \emph{IEEE Trans. Inf. Theory},
  vol.~48, no.~5, pp. 1118--1134, May 2002.

\bibitem{verdu02}
S.~Verd\'u, ``Spectral efficiency in the wideband regime,'' \emph{IEEE Trans.
  Inf. Theory}, vol.~48, no.~6, pp. 1319--1343, June 2002.

\bibitem{wienermasani57}
N.~Wiener and P.~Masani, ``The prediction theory of multivariate stochastic
  processes {I},'' \emph{Acta Math.}, vol.~98, pp. 111--150, Nov. 1957.

\bibitem{gradshteynryzhik00}
I.~S. Gradshteyn and I.~M. Ryzhik, \emph{Table of Integrals, Series, and
  Products}, 6th~ed., A.~Jeffrey, Ed.\hskip 1em plus 0.5em minus 0.4em\relax
  Academic Press, San Diego, 2000.

\bibitem{kramer07}
G.~Kramer, \emph{Topics in Multi-User Information Theory}.\hskip 1em plus 0.5em
  minus 0.4em\relax Foundations and Trends in Communication and Information
  Theory, 2007, vol.~4, no. 4--5, pp. 265--444.

\bibitem{gallager68}
R.~G. Gallager, \emph{Information Theory and Reliable Communication}.\hskip 1em
  plus 0.5em minus 0.4em\relax John Wiley \& Sons, 1968.

\bibitem{rimoldiurbanke96}
B.~Rimoldi and R.~Urbanke, ``A rate-splitting approach to the Gaussian
  multiple-access channel,'' \emph{IEEE Trans. Inf. Theory}, vol.~42, no.~2,
  pp. 364--375, Mar. 1996.

\bibitem{algoetcover88}
P.~H. Algoet and T.~M. Cover, ``A sandwich proof of the
  {Shannon-McMillan-Breiman} theorem,'' \emph{Ann. Prob.}, vol.~16, no.~2, pp.
  899--909, 1988.

\end{thebibliography}

\IEEEtriggeratref{10}

\end{document}